\documentclass[12pt,letterpaper]{article}
\usepackage{amsfonts,amsmath,amsthm,amssymb}
\usepackage{geometry,xcolor,graphicx,booktabs,enumerate,appendix}
\graphicspath{{figure/}{./}}
\usepackage{setspace}
\usepackage[authoryear,round]{natbib}
\usepackage{etoolbox}
\usepackage{titlesec}
\usepackage[protrusion=true,expansion=true,final]{microtype}
\usepackage{hyperref}
\hypersetup{colorlinks=true,citecolor=black,linkcolor=black,urlcolor=black}

\titleformat{\section}[block]{\large\bfseries}{\thesection.}{0.6em}{}
\titleformat{\subsection}[block]{\normalsize\bfseries}{\thesubsection.}{0.6em}{}
\titleformat{\paragraph}[runin]{\normalsize\itshape}{}{0em}{}[.]
\titlespacing*{\paragraph}{0pt}{0.8\baselineskip}{0.5em}
\titlespacing*{\section}{0pt}{1.4\baselineskip}{0.6\baselineskip}
\titlespacing*{\subsection}{0pt}{1.0\baselineskip}{0.4\baselineskip}

\newtheoremstyle{jbesthm}{\topsep}{\topsep}{\itshape}{0pt}{\bfseries}{.}{0.5em}{}
\newtheoremstyle{jbesdef}{\topsep}{\topsep}{\normalfont}{0pt}{\bfseries}{.}{0.5em}{}
\theoremstyle{jbesthm}
\newtheorem{theorem}{Theorem}
\newtheorem{proposition}{Proposition}
\newtheorem{lemma}{Lemma}
\theoremstyle{jbesdef}
\newtheorem{assumption}{Assumption}
\newtheorem{remark}{Remark}

\newtheorem{example}{Example}
\renewenvironment{proof}{\par\noindent\textit{Proof.}\ \ignorespaces}{\hfill$\square$\par\medskip}
\newenvironment{proofof}[1]{\par\noindent\textit{Proof of #1.}\ \ignorespaces}{\hfill$\square$\par\medskip}
\newcommand{\ppart}[1]{\par\noindent\textit{Part (#1).}\ \ignorespaces}

\DeclareMathOperator{\var}{Var}
\DeclareMathOperator{\RV}{RV}

\newcommand{\dto}{\overset{d}{\to}}
\newcommand{\Prob}{\mathbb{P}}
\newcommand{\E}{\mathbb{E}}
\newcommand{\Real}{\mathbb{R}}

\begin{document}
\allowdisplaybreaks
\widowpenalty=1000
\clubpenalty=1000
\emergencystretch=1em
\setlength{\parfillskip}{0pt plus .7\textwidth}
\renewcommand{\topfraction}{0.9}
\renewcommand{\bottomfraction}{0.8}
\renewcommand{\textfraction}{0.07}
\renewcommand{\floatpagefraction}{0.75}
\doublespacing
\pagestyle{myheadings}
\markboth{\normalfont Ji, Liu, Wang, and Xing}{\normalfont A Diagnostic Test for Binary Choice Models}

\thispagestyle{empty}
\begin{center}
{\large\bfseries A Simple and Powerful Diagnostic Test for Binary Choice Models}\let\thefootnote\relax\footnote{Ting Ji: jiting@cufe.edu.cn. Laura Liu: laura.liu@pitt.edu. Yulong Wang: yuw925@lehigh.edu. Jiahe Xing: jxing04@syr.edu. We thank Xu Cheng, Wayne Gao, Bo Honor\'e, Elena Manresa, Guillaume Pouliot, and participants of the African Meeting of the Econometric Society for helpful discussions. All remaining errors are ours.}\setcounter{footnote}{0}

\medskip

{\setstretch{1.0}
Ting Ji\\
\textit{School of Economics, Central University of Finance and Economics}

\smallskip

Laura Liu\\
\textit{Department of Economics, University of Pittsburgh}

\smallskip

Yulong Wang\\
\textit{Department of Economics, Lehigh University}

\smallskip

Jiahe Xing\\
\textit{Department of Economics, Syracuse University}

\medskip

\today
}
\end{center}

\bigskip

{\setstretch{1.0}
\noindent\textbf{ABSTRACT}

\medskip

\noindent Conventional binary choice models, such as probit and logit, impose thin-tailed errors, and that tail determines whether the parameters of a binary choice model can be estimated at the regular rate.
We test the restriction on observables, asking whether the conditional choice probability decays at a polynomial rate in a covariate.
Identification of tail heaviness requires no independence, no linear index, and no homoskedasticity.
The test is simple to implement and attains nearly the point-optimal power envelope among invariant tests.
An application to firm innovation decisions rejects the thin tail.

\medskip

\noindent\textbf{KEYWORDS:} Testing; binary choice; extreme value theory; tail index.
}

\bigskip

\section{Introduction}\label{sec:intro}

Binary choice models are workhorses of empirical economics.
They are widely used to study decisions such as labor force participation, educational choice, export decision, and innovation behavior.
Since the pioneering work of \citet{McFadden1974} and \citet{Chamberlain1980}, probit and logit have become standard tools in empirical work, and their appeal comes from a simple index structure and the resulting tractability of estimation.
Yet their credibility depends not only on the specification of the observable index, but also on the assumed distribution of the latent error term, which probit and logit take to be thin-tailed.
In many economic settings rare but consequential shocks may be important for the decision itself, so whether the latent error is thin-tailed is a substantive restriction rather than a harmless assumption.

The restriction matters for model fit, estimation, and inference.
\citet{khan2010irregular} show that in binary choice and related limited dependent variable models the parameter is identified on a thin set, so the rate of convergence depends on the tail of the error relative to that of the regressor.
The special regressor estimator of \citet{Lewbel1997,Lewbel2000} is the leading case, and Proposition \ref{prop:consequences} locates the boundaries at which the estimator's RMSE becomes slower than root-$n$ and the estimand ceases to be well defined.
With two periods and unrestricted heterogeneity the information bound for the common parameter is zero unless the error is logistic \citep{Chamberlain2010}, so evidence against the tail class containing the logistic is evidence that no regular estimator exists.

To see the idea, consider the simple binary choice model
\begin{equation}\label{eq:baseline}
Y_i=I(X_i-\varepsilon_i\ge0),
\end{equation}
where $I(\cdot)$ denotes the indicator function, $X_i$ is an observed covariate, and $\varepsilon_i$ is an unobserved error.
The tail of the latent error is reflected in the tail of the observed covariate once we condition on the outcome, since among observations with $Y_i=0$ large values of $X_i$ must be matched by even larger realizations of $\varepsilon_i$.
The right tail of $X_i\mid Y_i=0$ therefore reflects the right tail of the latent error, and symmetrically the left tail of $X_i\mid Y_i=1$ reflects its left tail.

We state the hypothesis on an observable rather than on the latent error, and this is what makes the diagnostic feasible.
Let
\begin{equation}\label{eq:q}
q(x)=\Prob(Y_i=0\mid X_i=x)
\end{equation}
be the conditional choice probability, and consider the rate at which $q(x)\to0$.
We focus on two regimes under weak regularity: $q$ is regularly varying (RV), so that $q(x)$ decays like $x^{-1/\xi}$ for some $\xi>0$, or $q$ is rapidly varying, so that $q$ decays faster than every power of $x$.
In \eqref{eq:baseline} with $\varepsilon_i$ independent of $X_i$, the second regime is $F_\varepsilon$ in the Gumbel domain of attraction, which contains the normal and the logistic and hence probit and logit.
Because the decay rate of $q$ restricts only the joint law of $(Y_i,X_i)$, it requires no linear index, no homoskedasticity, and no restriction on the joint distribution of covariates and unobserved individual effects in the panel data setup.

Let $\lambda>0$ be the extreme value index of $F_X$, $\gamma$ that of $F_{X\mid Y=0}$, and $\xi$ that of the distribution of $\varepsilon$.
Theorem \ref{thm:ident} establishes that
\[
\gamma=\frac{\lambda\xi}{\lambda+\xi}\quad\text{if }q\in\RV_{-1/\xi},
\quad\quad
\gamma=0\quad\text{if }q\text{ is rapidly varying},
\]
so the decay regime is recovered from the conditional extremes of an observable, and testing whether the latent error is thin-tailed is equivalent to testing whether $\gamma=0$.

Based on this equivalence we construct a test from the normalized largest $k$ order statistics of a single heavy-tailed covariate in the subsample $X_i\mid Y_i=0$.
Inference uses the fixed-$k$ limit theory of \citet{muller2017fixed}, under which the self-normalized vector of those order statistics converges to a law indexed by $\gamma$ alone.
The weighted average power test of \citet{andrews1994optimal} applied to that experiment is optimal among invariant tests.
Proposition \ref{prop:power} characterizes its power exactly.
The test nearly attains the point-optimal envelope.
In addition, Theorem \ref{thm:uniform} in the appendix shows that no uniform size statement is available over the unrestricted null and that uniformity is restored on a second-order restricted subclass.

The diagnostic is attractive in practice.
It requires neither model estimation nor a parametric specification of the link, and it is informative about the RMSE rate of the special regressor estimator before any binary choice model is estimated.
When one covariate dominates the tail of the index, Proposition \ref{prop:multi} shows that the test can be implemented using that covariate alone, and Appendix \ref{app:plugin} gives a plug-in version that does not require knowing which covariate dominates.
The same idea applies to short panels with unobserved individual effects, and because we estimate neither the common coefficients nor the individual effects, the incidental parameters problem does not arise.

The diagnostic is also conceptually distinct from a test of the parametric specification itself, such as the information matrix test, whose power is generated where the data are dense because the likelihood weights the extreme observations negligibly.
Our statistic instead uses the largest $k$ observations only, which is the part of the sample that \citet{khan2010irregular} identify as rate determining, and our null is a tail class containing both the normal and the logistic rather than a single parametric family.

Monte Carlo simulations show that the proposed test controls size under thin-tailed errors and has power against heavy-tailed alternatives close to the exact power envelope.
We then apply it to firms' innovation decisions, in which lagged value added serves as the heavy-tailed covariate, and the data provide evidence against the thin-tail assumption underlying conventional probit or logit.

\paragraph{Related literature}
The econometric analysis of binary outcomes has a long tradition. See \citet{McFadden1974}, \citet{Chamberlain1980}, and the textbook treatment of \citet{Wooldridge2002}.
A substantial literature studies specification and goodness of fit, asking whether the chosen link is appropriate, whether the parametric form is correct, or whether the implied choice probabilities match the data. See \citet{Pregibon1980}, \citet{Bierens1990}, \citet{HorowitzHardle1994}, \citet{Xia2004}, and, more recently, \citet{OtaOtsu2026}.
This literature does not directly test the tail behavior of the latent error, and our paper fills that gap.

The natural comparison is a test of the parametric specification itself, the leading example being the information matrix test of \citet{White1982}; see also \citet{Newey1985} and \citet{ChesherSpady1991}.
Unlike our tail-class diagnostic, these procedures assess a specified parametric family and draw their power primarily from observations in the mid-sample. 
Remark \ref{rem:moments} explains the underlying likelihood weights, and Section \ref{sec:simu:cross-sectional} compares the procedures in finite samples.

A related strand develops semiparametric estimators, beginning with the maximum score estimator of \citet{Manski1975} and followed by \citet{Han1987}, \citet{Horowitz1992}, \citet{Ichimura1993}, and \citet{KleinSpady1993}.
These relax the parametric form of the link but are estimation procedures rather than diagnostics.
Closest to our setting are the special regressor estimators of \citet{Lewbel1997,Lewbel2000}, which identify the intercept from an inverse weighted average over a covariate with large support, and \citet{khan2010irregular}, who show that the parameters of such models are identified on thin sets. \citet{AndrewsSchafgans1998} obtain the analogue for a sample selection model.
A rejection indicates that the conventional logit and probit specifications may not be adequate and points toward these estimators.
In panel binary choice models fixed effects generate the incidental parameters problem \citep{NeymanScott1948}, compounded in dynamic models by lagged dependence and initial conditions \citep{Heckman1981}. The bias correction of \citet{FernandezVal2009} and the estimator of \citet{HonoreKyriazidou2000} address these problems while retaining the thin-tailed structure of probit or logit, whereas \citet{Manski1987} develops a semiparametric approach based on conditional median restrictions rather than a specified parametric link.

Finally, we draw on extreme value theory in econometrics.
Classical references such as \citet{deHaan07} show that tail behavior determines the limiting laws of extremes, \citet{muller2017fixed} develop the fixed-$k$ framework that allows inference from a finite number of extreme order statistics, and \citet{einmahl2023extreme} provide results on tail index estimation under heterogeneity.
Standard extreme value methods study the tail of an observed variable directly, whereas we use the conditional extremes of an observed covariate to infer the tail class of an unobserved error.
In contrast to \citet{liu2025binary}, which considers estimation and forecasting with binary outcomes and extreme covariates, the contribution here is a test, and it accommodates the Gumbel domain, which is the thin-tailed null.

\paragraph{Organization}
Section \ref{sec:estimand} states what the decay rate determines for estimation and inference, Section \ref{sec:ident} establishes identification of tail heaviness and treats multiple covariates, and Section \ref{sec:inference} develops the test and its power.
Sections \ref{sec:mc} and \ref{sec:emp} report the simulations and the application, and Section \ref{sec:conclusion} concludes.
Proofs and additional results are in the Appendix.

\paragraph{Notation}
$\RV_{\theta}$ is the class of functions regularly varying at infinity with index $\theta$, so $h\in\RV_{\theta}$ means $h(tx)/h(x)\to t^{\theta}$ for every $t>0$; $h$ is slowly varying when $\theta=0$ and rapidly varying when $h(tx)/h(x)\to0$ for every $t>1$.
A positive measurable function $h$ on $(z_0,\infty)$ is $\Gamma$ varying with auxiliary function $a$ if $h\{t+xa(t)\}/h(t)\to e^{-x}$ for every $x\in\Real$; see \citet[Section 1.2]{deHaan07}.
$A\sim B$ means $A/B\to1$, and $A\asymp B$ means that $A/B$ is bounded above and below by positive constants for large arguments.

\section{The Estimand and What It Determines}\label{sec:estimand}

Before introducing our test, we first explain why the tail of the latent error is an economic object and not a technical one.
The special regressor estimator of \citet{Lewbel1997,Lewbel2000} avoids committing to a link function by identifying the intercept through inverse weighting, and the weight is largest exactly where the data are thinnest.
\citet{khan2010irregular} show that the parameter is identified on a thin set, so the rate of convergence is determined by how rare the informative observations are, and this depends on the tail of the error relative to the tail of the regressor.
The tail is therefore not a nuisance: it distinguishes among a root-$n$ RMSE rate, a slower polynomial RMSE rate, and a case in which the estimand is not well defined.

Consider that model with the covariate of this paper as the special regressor,
\begin{equation}\label{eq:lewbel}
Y_i=I\left(\theta_0+X_i-\varepsilon_i\ge0\right),\quad X_i\perp\varepsilon_i,\quad \E[\varepsilon_i]=0,
\end{equation}
where $F_X$ has infinite left and right endpoints and density $f_X$.
The intercept is identified by $\theta_0=\E[\{Y_i-I(X_i>0)\}/f_X(X_i)]$, which commits to no link function.
\citet{khan2010irregular} show that the second moment and the bias of the trimmed estimator
\begin{equation}\label{eq:trim}
\widehat\theta_n=\frac1n\sum_{i=1}^{n}\frac{Y_i-I(X_i>0)}{f_X(X_i)}I(|X_i|\le c_n)
\end{equation}
are governed by
\begin{equation}\label{eq:ktvb}
v(c_n)=\int_{-c_n}^{c_n}\frac{p(u)I(u\le0)+q(u)I(u>0)}{f_X(u)}du,
\quad
b(c_n)=\int_{c_n}^{\infty}q(u)du-\int_{-\infty}^{-c_n}p(u)du,
\end{equation}
with $p=1-q$ and $q$ as in \eqref{eq:q},
and that the efficiency bound is infinite as a worst case over unrestricted tails.
The integrand of $v$ is the object of this paper.
The economic content of \eqref{eq:ktvb} is simple.
Observations with $X_i$ far from zero are the ones that identify $\theta_0$, because for them the outcome is nearly determined and the sign of $Y_i-I(X_i>0)$ is informative.
Such observations are rare, and $q(u)$ measures how rare.
If $q$ vanishes slowly relative to $f_X$, these observations are relatively abundant, the variance $v(c_n)$ grows quickly with the trimming bound, and the estimator has a slower polynomial RMSE rate. If $q$ vanishes sufficiently quickly relative to $f_X$, they are scarce, but the weight $1/f_X$ does not destabilize the average, and the estimator has a root-$n$ RMSE rate. The bias behaves in the opposite direction, since trimming removes exactly those observations, so, in the slower-rate region and absent cancellation between the leading right- and left-tail bias terms, the trimming sequence balances the two at the polynomial-order level. Restricting the tails to the regularly varying class yields closed-form expressions for the variance order and a bias bound, which we use in the following proposition.

\begin{proposition}\label{prop:consequences}
Suppose $f_X$ is positive on $\Real$, $1/f_X$ is locally integrable, $1-F_X\in\RV_{-1/\lambda}$ and $f_X\in\RV_{-1/\lambda-1}$ with $\lambda>0$, and $q\in\RV_{-1/\xi}$ with $\xi\in(0,1)$.
Write $\gamma=\lambda\xi/(\lambda+\xi)$, so that $0<\gamma<\frac{\lambda}{1+\lambda}$.
Suppose the binary choice model \eqref{eq:lewbel} holds, that $\widehat\theta_n$ is the trimmed estimator \eqref{eq:trim}, and that $F_X$ and $\Prob(Y_i=1\mid X_i=\cdot)$ satisfy in the left tail the conditions imposed above on the right tail, with the same indices $\lambda$ and $\xi$.
Then for any trimming sequence $c_n=n^{\gamma+o(1)}$,
\begin{equation}\label{eq:rate}
\left[\E\left\{\left(\widehat\theta_n-\theta_0\right)^{2}\right\}\right]^{1/2}=n^{-r+o(1)},
\quad r=\min\left\{\tfrac12,\ 1-\gamma\left(1+\lambda^{-1}\right)\right\},
\end{equation}
where $r=1/2$ if and only if $\gamma\le\lambda/\{2(1+\lambda)\}$ and $0<r<1/2$ otherwise.
\end{proposition}

Thus $r=1/2$ gives a root-$n$ RMSE rate up to subpolynomial factors, whereas $0<r<1/2$ gives a slower polynomial RMSE rate. In both cases, the estimator remains consistent.

Every threshold in Proposition \ref{prop:consequences} is a function of the pair $(\lambda,\gamma)$ alone.
Figure \ref{fig:regimes} partitions that pair into three regions.
In the first $\widehat\theta_n$ attains the parametric rate, in the second it converges more slowly, and in the third the estimand is unavailable.
The functional $\int\{p(u)-I(u>0)\}du$ that $\widehat\theta_n$ targets converges if and only if $\E|\varepsilon_i|<\infty$, which fails for every $\xi>1$, equivalently for every $\gamma>\lambda/(1+\lambda)$.
Proposition \ref{prop:consequences} is stated for $\xi<1$ because at $\xi=1$ regular variation alone does not settle integrability, and because the exponent in \eqref{eq:rate} is zero there in any case.
An applied researcher who intends to use the special regressor estimator, or who needs to know whether a reported standard error is informative, must know which of these regions the design lies in.

One coordinate is easy, since $X_i$ is observed and $\lambda$ can be estimated from it directly.
The other is not, because $\gamma$ depends on the tail of $\varepsilon_i$, which is latent.
This is what motivates the rest of the paper, which shows that $\gamma$ is nevertheless recovered from the conditional extremes of $X_i$, so that the position in Figure \ref{fig:regimes} is determined without estimating the binary choice model and without choosing a link.
Before any binary choice model is estimated, the diagnostic provides the RMSE rate of the special regressor estimator and the trimming order $c_n=n^{\gamma+o(1)}$, which attains the exponent of Proposition \ref{prop:consequences}.

Proposition \ref{prop:consequences} requires $F_X$ to have an infinite left endpoint as well as an infinite right endpoint, since $\theta_0$ is identified from both tails.
That is a requirement of the special regressor estimator and not of the test, which uses one tail at a time.

\begin{figure}[tp]
\centering
\includegraphics[width=0.70\textwidth]{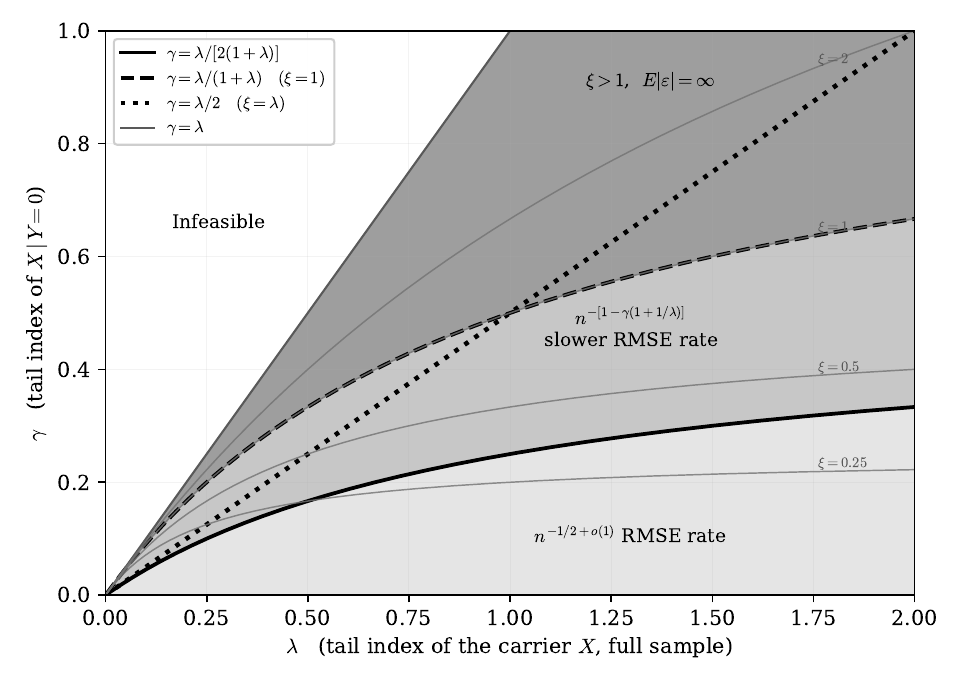}
\caption{Regime map in the identified pair $(\lambda,\gamma)$. Thin grey curves are contours of $\xi=\gamma\lambda/(\lambda-\gamma)$. The feasible region is $0<\gamma<\lambda$. The solid curve separates the root-$n$ and slower polynomial RMSE regions of Proposition \ref{prop:consequences}, and the dashed curve marks $\xi=1$, beyond which the special regressor estimand is not well defined.}
\label{fig:regimes}
\end{figure}

\begin{remark}[Moments and likelihood weights]\label{rem:moments}
Theorems 3.1 and A.1 of \citet{khan2010irregular} show that finite second moments of the regressors imply an infinite semiparametric efficiency bound in \eqref{eq:lewbel}.
If $\lambda>1/2$, then $\E[X_i^{2}]=\infty$, and such a covariate is therefore necessary rather than disqualifying for regular estimation.
The parametric likelihood is unaffected either way.
Consider the baseline model \eqref{eq:baseline} where $\varepsilon$ has distribution function $F_\varepsilon$ and density $f_\varepsilon$.
Its information matrix is $\E[X_i^2\omega(X_i)]$ with $\omega=f_\varepsilon^{2}/[F_\varepsilon(1-F_\varepsilon)]$, and, writing $\phi$ for the standard normal density, $\omega(t)\sim|t|\phi(t)$ under probit and $\omega(t)\sim e^{-|t|}$ under logit, so it is finite whenever the regressors have a proper distribution.

What the tail does change is the weight $\omega$ places on the extremes.
A test of the parametric specification itself, such as the information matrix test of \citet{White1982}, is built from the score and the Hessian, both of which weight an observation with fitted index $t$ by a factor of order $t^{2}\omega(t)$, and $t^{2}\omega(t)\to0$ under both links.
The largest observations are asymptotically uninformative for such a statistic, whose power comes from the mid-sample, whereas $\varphi_k$ in Section \ref{sec:inference} is a function of the largest $k$ order statistics of one covariate in the subsample $Y_i=0$ and of nothing else.
The two read disjoint parts of the sample and are aimed at different objects.
A likelihood-based test asks whether a given parametric family fits where the data are dense.
The hypothesis of this paper asks whether $q$ decays at a polynomial rate, which under Proposition \ref{prop:consequences} determines whether the special regressor estimator has a root-$n$ or slower polynomial RMSE rate.
\end{remark}

\section{Identification of Tail Heaviness}\label{sec:ident}

In this section we first show how the tail heaviness of the latent error is recovered from the covariate, and then record which structural models generate the condition under which that recovery is valid.
Section \ref{sec:ass} states the assumption and Section \ref{sec:thm} the identification result.
Section \ref{sec:primitive} shows how index transformations and heteroskedasticity enter, and Section \ref{sec:multi} treats the case of multiple covariates.

\subsection{Preliminaries and Assumption}\label{sec:ass}

To characterize the tail of the latent error, we consider the tail of $X_i$ through extreme value theory.
The generalized extreme value family is $G_{\gamma}(x)=\exp\{-(1+\gamma x)^{-1/\gamma}\}$ for $\gamma\ne0$ on $\{x:1+\gamma x>0\}$ and $G_0(x)=\exp\{-e^{-x}\}$.
For $X_i$ i.i.d.\ $F$, we write $F\in\mathcal D(G_{\gamma})$ if $(\max_{i\le n}X_i-b_n)/a_n\dto G_{\gamma}$ for some $a_n>0$ and $b_n$, and call $\gamma$ the extreme value index.
The case $\gamma=0$ is the Gumbel case and includes the normal, logistic, and exponential distributions, while $\gamma>0$ is the Fr\'echet case and is equivalent to $1-F\in\RV_{-1/\gamma}$ \citep[Theorem 1.2.1]{deHaan07}.

To illustrate the key idea, consider first a simple example in which $X_i$ follows a Pareto law and $\varepsilon_i$ follows either a Pareto or an exponential law, with $\varepsilon_i$ independent of $X_i$ and $Y_i$ generated by \eqref{eq:baseline}.
By Bayes' rule the conditional density is
\begin{equation}\label{eq:tail}
f_{X\mid Y=0}(x)=\frac{q(x)f_X(x)}{\Prob(Y_i=0)}=\frac{\{1-F_\varepsilon(x)\}f_X(x)}{\Prob(Y_i=0)} .
\end{equation}
Let $X_i$ be Pareto with index $\lambda>0$, so that $f_X(x)\propto x^{-1/\lambda-1}$.
When $\varepsilon_i$ is Pareto with index $\xi>0$ we have $1-F_\varepsilon(x)\propto x^{-1/\xi}$ and hence $f_{X\mid Y=0}(x)\propto x^{-1/\lambda-1/\xi-1}$, so that $F_{X\mid Y=0}\in\mathcal D(G_{\lambda\xi/(\lambda+\xi)})$.
When $\varepsilon_i$ is exponential with rate $\eta$ we have $1-F_\varepsilon(x)=e^{-\eta x}$ and hence $f_{X\mid Y=0}(x)\propto x^{-1/\lambda-1}e^{-\eta x}$, which decays exponentially, so that $F_{X\mid Y=0}\in\mathcal D(G_0)$.
Combining the two cases, $F_{X\mid Y=0}\in\mathcal D(G_\gamma)$ with $\gamma=\lambda\xi/(\lambda+\xi)$, and the exponential case gives $\gamma=0$.

To generalize beyond the Pareto and exponential cases, and beyond the independence between $\varepsilon_i$ and $X_i$ that \eqref{eq:tail} uses, we impose the following conditions on the joint distribution of $(Y_i,X_i)$.

\begin{assumption}[Tail Conditions]\label{ass:main}
Let the following statements hold:
\begin{enumerate}[(i).]
\item $\Prob(Y_i=0)\in (0,1)$.
\item $F_X$ has right endpoint $+\infty$ and a positive, continuously differentiable density $f_X$ on $(z_0,\infty)$, with $\lim_{x\to\infty}xf_X'(x)/f_X(x)=-(1/\lambda+1)$ for some $\lambda>0$.
\item $q$ in \eqref{eq:q} is positive and continuous on $(z_0,\infty)$, and exactly one of the following holds: (a) $q\in\RV_{-1/\xi}$ for some $\xi\in(0,\infty)$; (b) $q$ is $\Gamma$ varying; (c) $q\in\RV_{0}$.
\end{enumerate}
\end{assumption}

Here and below, $z_0$ denotes a finite threshold that may differ across conditions.

We now discuss Assumption \ref{ass:main}. 
Part (i) is standard, requiring that the unconditional choice probability is non-degenerate. 
Parts (ii) and (iii.a) supply the right tail hypotheses of Proposition \ref{prop:consequences}.
In particular, Part (ii) is the von Mises condition for the Fr\'echet domain written on the density.
It implies $f_X\in\RV_{-1/\lambda-1}$ and, by Karamata's theorem, $1-F_X\sim\lambda xf_X(x)\in\RV_{-1/\lambda}$, so $F_X\in\mathcal D(G_\lambda)$.
Part (iii) places $q$ in one of three tail classes.
The tail index is respectively $-1/\xi$ under (a), $-\infty$ under (b) since a $\Gamma$ varying function is rapidly varying, and zero under (c), so the cases are mutually exclusive.
Case (a) is $q\in\RV_{-1/\xi}$, so the choice probability vanishes at the polynomial rate $x^{-1/\xi}$.
Case (b) is the class of \citet{deHaan07} associated with the Gumbel domain, in which $q$ vanishes faster than any power at a rate set by the auxiliary function.
It holds whenever $q$ is twice continuously differentiable with $q'<0$ and its reciprocal hazard $a_q=-q/q'$ satisfies $a_q'(x)\to0$, the von Mises condition, which is how Section \ref{sec:primitive} and Table \ref{tab:examples} verify it in practice.
Case (c) is $q$ slowly varying, so the covariate does not shift the choice probability at a polynomial rate.
It admits a constant $q$, which is the configuration of the last row of Table \ref{tab:examples}.

Assumption \ref{ass:main} is stated for the right tail.
The left tail statistic is the right tail statistic computed from the symmetric pair $(1-Y_i,-X_i)$.
Its choice probability is $q^{-}(x)=\Prob(Y_i=1\mid X_i=-x)$.
The statistic is valid under Assumption \ref{ass:main} applied to that pair, which we label Assumption \ref{ass:main}$^{-}$.
The same formulas deliver a left tail index $\gamma^{-}$.
The two conditions are separate restrictions, the first on the decay of $\Prob(Y_i=0\mid X_i=x)$ as $x\to+\infty$ and the second on the decay of $\Prob(Y_i=1\mid X_i=x)$ as $x\to-\infty$, and the combined procedure of Section \ref{sec:test} tests their conjunction $\gamma=\gamma^{-}=0$.
For a covariate bounded below, such as firm assets, only Assumption \ref{ass:main} is available.

We note that Assumption \ref{ass:main} restricts only the joint distribution of $(Y_i,X_i)$.
It does not require the latent error to be independent of the covariates, nor impose linear index form or homoskedasticity. 
As we discuss in Remark \ref{rem:panel}, it does not impose any restriction on the joint distribution of covariates and unobserved individual effects either in panel data models.
Section \ref{sec:primitive} records which structural models generate it and their corresponding values of $\xi$ and $\gamma$.

\subsection{The Identification Result}\label{sec:thm}

It is not straightforward to test the decay rate of $q$ directly, since $q$ is a conditional probability in a region where observations are scarce by construction and $\varepsilon_i$ is latent.
The next result shows that the rate is nevertheless encoded in the extreme value index of $F_{X\mid Y=0}$, which is the law of an observable in an observed subsample.
We focus on the right tail, and the same argument applies to the left. 
Recall that $\gamma$ denotes the extreme value index of
$F_{X\mid Y=0}$, $\lambda$ denotes that of $F_X$, and $\xi$ denotes
the regular-variation index of $q$ in Assumption
\ref{ass:main}(iii.a).

\begin{theorem}\label{thm:ident}
Suppose Assumption \ref{ass:main}(i) holds.  
\begin{enumerate}[(a).]
\item Suppose Assumption \ref{ass:main}(ii) and (iii.a) hold.
Then
\begin{equation}\label{eq:gamma}
1-F_{X\mid Y=0}\in\RV_{-(1/\lambda+1/\xi)},
\quad
\gamma=\frac{\lambda\xi}{\lambda+\xi}\in(0,\lambda),
\quad
\xi=\frac{\gamma\lambda}{\lambda-\gamma}.
\end{equation}
\item Suppose Assumption \ref{ass:main}(ii) and (iii.b) hold.
Then $F_{X\mid Y=0}\in\mathcal D(G_0)$, so that $\gamma=0$, and $\{1-F_{X\mid Y=0}\}/f_{X\mid Y=0}\sim b_q$ for any auxiliary function $b_q$ of $q$.
\item Suppose Assumption \ref{ass:main}(ii) and (iii.c) hold.
Then $1-F_{X\mid Y=0}\in\RV_{-1/\lambda}$ and $\gamma=\lambda$.
\item Suppose $f_X$ is positive and twice continuously differentiable on $(z_0,\infty)$ with $f_X'<0$, and that $a_X=-f_X/f_X'$ satisfies $a_X'(x)\to0$, so that Assumption \ref{ass:main}(ii) holds for no $\lambda>0$.
Then $\gamma=0$ and $F_{X\mid Y=0}\in\mathcal D(G_0)$ under parts (iii.a) and (iii.c), and under part (iii.b) whenever $q$ has an auxiliary function $b_q$ such that $b_q/a_X$ converges in $[0,\infty]$ or $q$ is twice continuously differentiable with $q'<0$ and $a_q'\to0$.
\end{enumerate}
\end{theorem}

Parts (a) and (b) are the identification result.
Both $\lambda$ and $\gamma$ are extreme value indices of observable
distributions, the first of $F_X$ in the full sample and the second of
$F_{X\mid Y=0}$ in the subsample with $Y_i=0$, so
\eqref{eq:gamma} identifies the decay index of $q$ from two estimable
quantities. Under the baseline model, this is the tail index of the
latent error.
Part (a) delivers $\gamma>0$ and part (b) delivers $\gamma=0$, so
testing whether the error is thin-tailed is testing whether $\gamma=0$.

The identifying power arises from the extreme observations of $X_i$.
Since $Y_i=0$ implies $X_i<\varepsilon_i$ in \eqref{eq:baseline}, large values of $X_i$ in that subsample correspond to extreme realizations of the latent error, so a heavy-tailed regressor conveys tail information: its extreme order statistics determine the tail class of the unobserved error.
Requiring one heavy-tailed covariate is therefore not a technical convenience but the source of identification of the tail heaviness.

Parts (c) and (d) are the two ways in which identification of the tail heaviness fails, and each is the failure of one assumption.
Part (c) is case (iii.c), in which the covariate does not shift the choice probability at a polynomial rate, and the conditional and unconditional tails then coincide and $\gamma=\lambda$.
Because parts (a) and (b) place $\gamma$ strictly below $\lambda$, a confidence interval for $\gamma$ that is not separated from one for $\lambda$ signals this failure.
Part (d) is the failure of (ii), in which the covariate itself lies in the Gumbel domain.
Under the conditions of part (d), a thin-tailed covariate delivers $\gamma=0$, and the sample rarely generates observations large enough to reveal the tail class.

\begin{remark}[Levels, not logarithms]\label{rem:levels}
Part (d) states that a covariate in the Gumbel domain with $\lambda=0$ delivers $\gamma=0$, so no procedure based on its conditional extremes can distinguish the two regimes.
The same conclusion is visible in \eqref{eq:gamma}, where $\gamma\to0$ as $\lambda\to0$ for every fixed $\xi$.
Therefore, if a covariate has a Pareto tail, it must enter the model in levels but not in logarithms, since the logarithm of a Pareto variable is exponential and carries no identifying content.
The consequences of a logarithmic index for the interpretation of a rejection are given in Table \ref{tab:examples}.
\end{remark}

\subsection{More General Models}\label{sec:primitive}

Assumption \ref{ass:main} restricts $q$ and $F_X$ rather than a structural model, which is what makes Theorem \ref{thm:ident} robust.
This section asks the converse question.
We consider a class of binary choice models that allows a nonlinear index and heteroskedasticity of unknown form.
For each member we determine which part of Assumption \ref{ass:main}(iii) holds, and, when part (iii.a) holds, the values of $\gamma$ and $\xi$ it holds with.
A single formula, \eqref{eq:aq} below, answers both questions, and Table \ref{tab:examples} applies it to the specifications used in practice.

Let
\begin{equation}\label{eq:general}
Y_i=I\left\{m(X_i)\ge\sigma(X_i)\varepsilon_i\right\},\quad \varepsilon_i\perp X_i .
\end{equation}
Here $m$ is the index function, differentiable and strictly increasing with $m(x)\to\infty$, and $\sigma>0$ is the conditional scale of the error, also differentiable.
The unobserved error $\varepsilon_i$ has survival function $S(t)=\Prob(\varepsilon_i>t)$ with $S'<0$ and $f_S=-S'$ differentiable.
The baseline model \eqref{eq:baseline} is the case $m(x)=x$ and $\sigma\equiv1$.
A nonlinear index is the case $\sigma\equiv1$ with general $m$, and heteroskedasticity of unknown form is the case $m(x)=x$ with general $\sigma$.

Because $\sigma>0$, the events $\{m(X_i)\ge\sigma(X_i)\varepsilon_i\}$ and $\{m(X_i)/\sigma(X_i)\ge\varepsilon_i\}$ coincide, so independence gives
\begin{equation}\label{eq:comp}
q(x)=\Prob\left\{\varepsilon_i>\widetilde m(x)\right\}=S\{\widetilde m(x)\},
\quad
\widetilde m:=\frac{m}{\sigma} .
\end{equation}
We call $\widetilde m$ the effective index.
Equation \eqref{eq:comp} is an identity rather than a limit, and it says that the index function and the conditional scale enter the choice probability only through their ratio.
Nothing beyond the shape of $\widetilde m$ and the tail of $S$ has to be known to read the tail of $q$.

When $\widetilde m'>0$, differentiating \eqref{eq:comp} gives $q'=-f_S(\widetilde m)\widetilde m'<0$, so, with $a_S=S/f_S$ the reciprocal hazard of $S$,
\begin{equation}\label{eq:aq}
a_q(x)=-\frac{q(x)}{q'(x)}=\frac{a_S\{\widetilde m(x)\}}{\widetilde m'(x)} .
\end{equation}
The reciprocal hazard of the choice probability is the reciprocal hazard of the error, evaluated at the effective index and divided by its slope.
Every question about Assumption \ref{ass:main}(iii.a) or (iii.b) is therefore a question about $a_S$ and $\widetilde m$.
Part (iii.b) is implied by the von Mises condition $a_q'\to0$, which is settled by differentiating \eqref{eq:aq} and is the route taken in every case below.
Part (iii.a) asks whether $q$ is regularly varying with a finite index, and for a regularly varying error \eqref{eq:aq} answers it in closed form through the monotone density theorem.

Suppose $S\in\RV_{-1/\xi_S}$ with $f_S$ eventually monotone, and $\widetilde m\in\RV_{\kappa}$ with $\widetilde m'$ eventually monotone, where $\xi_S>0$ and $\kappa>0$.
Regular variation restricts the level of a function and not that of its derivative, so the passage from $S$ to $a_S$ uses the monotone density theorem \citep[Theorem 1.7.2]{Bingham1987} rather than the Karamata representation.
Applied to $S(t)=\int_t^{\infty}f_S$ it gives $a_S(t)/t\to\xi_S$, and applied to $\widetilde m$ it gives $x\widetilde m'(x)/\widetilde m(x)\to\kappa$.
Since $\widetilde m(x)\to\infty$, \eqref{eq:aq} factors as
\begin{equation}\label{eq:xifactor}
\frac{a_q(x)}{x}=\frac{a_S\{\widetilde m(x)\}}{\widetilde m(x)}\cdot\frac{\widetilde m(x)}{x\widetilde m'(x)}\longrightarrow\frac{\xi_S}{\kappa} ,
\end{equation}
so Assumption \ref{ass:main}(iii.a) holds with $\xi=\xi_S/\kappa$.
The tail index of the error is divided by $\kappa$, the index of regular variation of the effective index.

Four cases arise, and they exhaust Table \ref{tab:examples}.
If $\widetilde m\in\RV_\kappa$ with $\kappa>0$ and the error is heavy-tailed, part (iii.a) holds with $\xi=\xi_S/\kappa$ and $\gamma>0$.
If $\widetilde m\in\RV_\kappa$ with $\kappa>0$ and the error is in the Gumbel class, part (iii.b) holds and $\gamma=0$, so a power index leaves the two classes intact and only rescales $\xi$.
If $\widetilde m$ grows more slowly than any power, \eqref{eq:xifactor} does not apply and \eqref{eq:aq} must be used directly, and a logarithmic index then converts a Gumbel error into part (iii.a) with a finite $\xi$.
If $\sigma$ grows faster than $m$, then $\widetilde m\to0$ and $q(x)\to S(0)>0$.
Hence $q\in\RV_{0}$, and Theorem \ref{thm:ident}(c) applies directly without using \eqref{eq:aq}.

\begin{table}[tp]
\centering
\small
\resizebox{\textwidth}{!}{
\begin{tabular}{llllcl}
\toprule
Index $\widetilde m(x)$ & Error survival & $q(x)$ & $a_q(x)$ & Case, index & $\gamma$ \\
\midrule
$x$ & $a_S(t)/t\to\xi_S>0$ & $\in\RV_{-1/\xi_S}$ & $\sim\xi_S x$ & (a), $\xi_S$ & $\lambda\xi_S/(\lambda+\xi_S)$\\
$x$ & $S$ logistic & $\sim e^{-x}$ & $\to1$ & (b) & $0$\\
$x$ & $S$ normal & $\sim\phi(x)/x$ & $\sim1/x$ & (b) & $0$\\
$x^{\kappa}$, $\kappa>0$ & $a_S(t)/t\to\xi_S>0$ & $\in\RV_{-\kappa/\xi_S}$ & $\sim(\xi_S/\kappa)x$ & (a), $\xi_S/\kappa$ & $\lambda\xi_S/(\lambda\kappa+\xi_S)$\\
$x^{\kappa}$, $\kappa>0$ & $S$ logistic & $\sim e^{-x^{\kappa}}$ & $\sim x^{1-\kappa}/\kappa$ & (b) & $0$\\
$x^{\kappa}$, $\kappa>0$ & $S$ normal & $\asymp x^{-\kappa}e^{-x^{2\kappa}/2}$ & $\sim x^{1-2\kappa}/\kappa$ & (b) & $0$\\
$\tau\log x$, $\tau>0$ & $S$ logistic & $\sim x^{-\tau}$ & $\sim x/\tau$ & (a), $1/\tau$ & $\lambda/(1+\lambda\tau)$\\
$\tau\log x$, $\tau>0$ & $S$ normal & $\asymp\frac{e^{-\tau^{2}(\log x)^{2}/2}}{\log x}$ & $\sim\frac{x}{\tau^{2}\log x}$ & (b) & $0$\\
$x/\sigma(x)$, $\sigma\in\RV_{s}$, $s<1$ & $a_S(t)/t\to\xi_S>0$ & $\in\RV_{-(1-s)/\xi_S}$ & $\sim\frac{\xi_S x}{1-s}$ & (a), $\frac{\xi_S}{1-s}$ & $\frac{\lambda\xi_S}{\lambda(1-s)+\xi_S}$\\
$x/\sigma(x)$, $\sigma\in\RV_{s}$, $s>1$ & $S$ continuous at $0$ with $S(0)>0$ & $\to S(0)>0$ & - & $q\in\RV_{0}$ & $\lambda$\\
\bottomrule
\end{tabular}}
\caption{Special cases of \eqref{eq:general} and the implied conditional index.
Rows one to eight take $\sigma\equiv1$, so that $\widetilde m=m$, and rows nine and ten take $m(x)=x$, so that $\widetilde m=x/\sigma(x)$.
The second column restricts the survival function $S$ of $\varepsilon_i$, and the fifth gives the part of Assumption \ref{ass:main}(iii) that holds together with the value of $\xi$ in case (a).
Every entry in the fourth column follows from \eqref{eq:aq}, except the last row, for which $a_q$ is undefined.
Every entry in the last column follows from Theorem \ref{thm:ident}, the last row through part (c) because $q\in\RV_{0}$.}
\label{tab:examples}
\end{table}

Two implications of Table \ref{tab:examples} matter for practice.
First, rows four through six show that the diagnostic does not require the covariate to enter linearly: for any power index with $\kappa>0$, $\gamma=0$ if and only if the error tail is in the Gumbel class.
Second, rows seven and eight give the exact content of a rejection when the applied specification uses a logarithm.
A rejection obtained with the covariate in levels admits two readings.
One is a heavy-tailed error under a level index.
The other is an exponentially tailed error under a logarithmic index.
In either case, a probit or logit model written in the level of the covariate is misspecified.
A logarithmic index with a Gaussian error is not rejected.

In summary, Table \ref{tab:examples} illustrates cases (a)--(c) in Theorem \ref{thm:ident}, and our proposed test is aimed at cases (a) and (b) alone.
These are the two cases in which the covariate moves the choice probability at a polynomial rate or faster, and they are the empirically relevant ones.
Many standard binary choice specifications fall into one of these cases: a probit or logit in the level of the covariate is case (b), and any heavy-tailed error under the same index is case (a).
The null $\gamma=0$ is case (b) and the alternative $\gamma>0$ is case (a), so the hypothesis \eqref{eq:hypo} is a choice between them.
Case (c) is not an alternative the test is designed to detect.
It arises when the covariate does not move the index at all, either because the conditional scale absorbs it or because another regressor cancels it, and it is the configuration in which a rejection carries no information about $\varepsilon_i$.
Case (d) of Theorem \ref{thm:ident} applies when $X_i$ itself is thin-tailed, which we can easily check from the data. 
See Section \ref{sec:emp} for details. 
 
\subsection{Multiple Covariates}\label{sec:multi}

Our previous analysis focuses on a single dominating covariate. 
Now we consider the binary choice model with multiple covariates,
\begin{equation}\label{eq:multi}
Y_i=I\left(m(X_i)+X_i^{a\prime}\beta-\varepsilon_i\ge0\right),
\end{equation}
where $X_i$ is the scalar heavy-tailed dominating regressor, $X_i^{a}$ collects the auxiliary variables, and $m$ is the index function of \eqref{eq:general}.
The coefficient on the dominating regressor is normalized to one without loss of generality, provided that $X_i$ affects $Y_i$ conditional on $X_i^{a}$.
The question is what the auxiliary covariates do to the tail of $q$. 
Under the regularly varying alternative they have no asymptotic effect when their conditional moments are sufficiently light, while under a smooth Gumbel null their conditional distribution must stabilize in the relevant tail.

\begin{proposition}\label{prop:multi}
Suppose \eqref{eq:multi} holds with $\varepsilon_i\perp(X_i,X_i^{a})$.
Write $S_\varepsilon(t)=\Prob(\varepsilon_i>t)$ and $G_x$ for the conditional law of $X_i^{a\prime}\beta$ given $X_i=x$, so that
\begin{equation}\label{eq:mixture}
q(x)=\int S_\varepsilon\left\{m(x)+u\right\}dG_x(u).
\end{equation}
\begin{enumerate}[(a).]
\item If $S_\varepsilon\in\RV_{-1/\xi_S}$ with $\xi_S>0$, $m(x)\to\infty$ as $x\to\infty$, and
\begin{equation}\label{eq:condmom}
\sup_{x>z_0}\E\left[\left|X_i^{a\prime}\beta\right|^{p}\ \middle|\ X_i=x\right]<\infty
\quad\text{for some }p>1/\xi_S ,
\end{equation}
then $q\sim S_\varepsilon\circ m$.
Hence Assumption \ref{ass:main}(iii.a) holds with $\xi=\xi_S$ when $m(x)=x$ and with $\xi=\xi_S/\kappa$ when $m\in\RV_\kappa$, so that $\gamma=\lambda\xi/(\lambda+\xi)$.
\item Suppose $m$ is twice continuously differentiable and eventually strictly increasing, $m(x)\to\infty$ as $x\to\infty$, and, for some $\kappa>0$,
\begin{equation}\label{eq:msmooth}
\frac{xm'(x)}{m(x)}\longrightarrow\kappa,\quad
\frac{xm''(x)}{m'(x)}\longrightarrow\kappa-1.
\end{equation}
Let $U_i=X_i^{a\prime}\beta$ have bounded support $\mathcal U$, and suppose $G_x\ll G$ for a fixed probability $G$ on $\mathcal U$ and all large $x$, with $r_x=dG_x/dG$ twice continuously differentiable in $x$.
Suppose $S_\varepsilon$ is positive and twice continuously differentiable in the right tail with $S_\varepsilon'<0$, and write $a_\varepsilon=-S_\varepsilon/S_\varepsilon'$ and $A_m(x)=a_\varepsilon\{m(x)\}/m'(x)$.
If
\begin{align}\label{eq:gumbelshift}
&a_\varepsilon'(t)\longrightarrow0\quad(t\to\infty),\quad
\sup_{u\in\mathcal U}\left|\frac{a_\varepsilon\{m(x)+u\}}{a_\varepsilon\{m(x)\}}-1\right|\longrightarrow0\quad(x\to\infty),\\ \label{eq:tailstable}
&\sup_{u\in\mathcal U}\left\{|r_x(u)-1|+A_m(x)|\partial_x r_x(u)|+A_m(x)^2|\partial_x^2 r_x(u)|\right\}\longrightarrow0
\quad(x\to\infty),
\end{align}
then $q$ satisfies Assumption \ref{ass:main}(iii.b), so under Assumption \ref{ass:main}(i)-(ii), $F_{X\mid Y=0}\in\mathcal D(G_0)$ and $\gamma=0$.
\end{enumerate}
\end{proposition}

Proposition \ref{prop:multi} shows that in the presence of multiple covariates the tail behavior of the latent error can still be recovered from a single heavy-tailed regressor.
When one covariate dominates the right tail of the index, $F_{X\mid Y=0}$ remains in the same domain of attraction as in the univariate case, so by Theorem \ref{thm:ident} the conditional index is again that of the univariate problem, with $\xi$ the index of $q$.
A Gumbel error yields $\gamma=0$ by part (b).
A regularly varying error with index $\xi_S$ yields reduced-form index $\xi=\xi_S/\kappa$, so that $\gamma=\lambda\xi/(\lambda+\xi)>0$ by part (a).
The diagnostic can therefore be implemented using only the regressor with the heaviest tail, without estimating $\beta$ and without forming the index $m(X_i)+X_i^{a\prime}\beta$.

Proposition \ref{prop:multi} imposes no independence between $X_i$ and $X_i^{a}$, and it shows what dependence costs and what it does not.
When the error is heavy-tailed, correlation between $X_i$ and $X_i^{a}$ has no asymptotic effect provided the conditional law of the auxiliary index remains tight as $X_i$ diverges, which is what \eqref{eq:condmom} enforces through a conditional moment bound of order exceeding $1/\xi_S$.
The reason is that regular variation is insensitive to shifts of bounded order, so a mixture over a conditionally tight shift leaves the index of regular variation unchanged.
Condition \eqref{eq:condmom} holds for every $p$ when the components of $X_i^{a}$ are finite-support discrete or bounded, which is the case in the application below, and it is then not binding.
When some components have unbounded support, \eqref{eq:condmom} restricts how heavy their tails may be relative to that of $\varepsilon_i$, and it is the condition to be argued in an application.

\begin{example}[Firm innovation]\label{ex:innovation}
In models of firm innovation, R\&D involves fixed costs, so the returns to innovation increase with the scale over which an innovation can be applied, generating a positive relationship between firm scale and innovative activity \citep{cohenRepriseSize1996a,cohen2010fifty}. For a given year $t$, let $Y_i=1$ if firm $i$ files at least one patent. The firm's value added in year $t-1$ captures its pre-existing operating scale and serves as the dominating regressor $X_i$.
The auxiliary block $X_i^{a}$ includes ownership indicators, policy variables, and industry and region fixed effects, each of which is finite-support discrete or bounded.
\end{example}

When the auxiliary block is independent of $X_i$ the analysis needs nothing new.
The event in \eqref{eq:multi} is $\{m(X_i)\ge\varepsilon_i-X_i^{a\prime}\beta\}$, so $q=S\circ m$ with $S$ the survival function of the composite error $\varepsilon_i-X_i^{a\prime}\beta$.
This is \eqref{eq:comp}, and Section \ref{sec:primitive} applies with $S$ in place of the survival function of $\varepsilon_i$.
The only question is the tail class of the composite.
A finite-support discrete or bounded auxiliary block leaves a smooth Gumbel-class error satisfying \eqref{eq:gumbelshift} in case (iii.b), since the composite is then a mixture of bounded shifts of $\varepsilon_i$, and normal and logistic errors are leading examples.
Probit and logit with the usual controls therefore satisfy the null, which is what the test requires of them.
The composite leaves the Gumbel class only when $X_i^{a\prime}\beta$ is itself heavy-tailed, and then the rejection is correct: a rapidly varying link in $X_i$ alone is misspecified in that design.

Part (a) of Proposition \ref{prop:multi} allows the auxiliary block to depend on $X_i$ under the regularly varying alternative, with \eqref{eq:condmom} ensuring that the auxiliary index is too thin to affect the tail index of $\varepsilon_i$.
Part (b) covers smooth Gumbel nulls under the error-tail condition \eqref{eq:gumbelshift} and the tail-stability condition \eqref{eq:tailstable}, and the class includes normal and logistic errors.
The scale in \eqref{eq:tailstable} is the reciprocal-hazard scale of $q$ in $x$, $A_m(x)=a_\varepsilon\{m(x)\}/m'(x)$.
Condition \eqref{eq:tailstable} also permits dependence between $X_i$ and $X_i^a$: their conditional relationship may vary throughout the body of the distribution, but the distribution of the bounded auxiliary index must stabilize in the relevant tail.
Together, the two parts justify using the single dominating covariate on both sides of the test without estimating $\beta$.

\begin{remark}[The plug-in index]\label{rem:plugin}
Appendix \ref{app:plugin} gives a plug-in approach that does not require knowing which regressor dominates the tail.
Consider the model $Y_i=I(\widetilde X_i'\delta-\varepsilon_i>0)$, where $\widetilde X_i=(X_i,X_i^a)'$ and $\delta=(1,\beta')'$.
Write the index as $W_i=\widetilde X_i'\delta$, estimate $\delta$ by a distribution-free method such as the maximum score estimator of \citet{Manski1975}, and form $\widehat W_i=\widetilde X_i'\widehat\delta$.
Proposition \ref{prop:plugin} shows that the test applied to $\widehat W_i$ has the same asymptotic size as if $\delta$ were known. 
We introduce this plug-in approach for completeness and use the dominating regressor in the rest of the paper for its simplicity.
\end{remark}

\begin{remark}[Panel data]\label{rem:panel}
The period-specific choice probability is $q_t(x)=\Prob(Y_{it}=0\mid X_{it}=x)$.
It satisfies Assumption \ref{ass:main}(iii) under primitive conditions exactly as in the cross section.
Proposition \ref{prop:panel}(a) records that Theorems \ref{thm:ident} and \ref{thm:size} then apply period by period, with no restriction on the joint law of the covariate and the individual effect.
That is the gain from stating the hypothesis on $q_t$ rather than on a latent error.
The individual effect and any lagged outcome enter $q_t$ and need not be modeled, because the test reads the tail of $q_t$ but not its level.
The incidental parameters problem therefore does not arise.
\end{remark}

\section{Inference}\label{sec:inference}

Theorem \ref{thm:ident} converts a hypothesis on the latent error into a hypothesis on the extreme value index of an observable, and this section constructs a test of it.
Section \ref{sec:test} derives the statistic and its null distribution, and Sections \ref{sec:power} to \ref{sec:uniform} study power, the choice of $k$, and uniform size.

\subsection{The Test}\label{sec:test}

We again illustrate the procedure using the right tail.
Theorem \ref{thm:ident} directly implies the implementation: under the null hypothesis the normalized spacings of the largest $k$ order statistics of $X_i\mid Y_i=0$ follow a fixed-$k$ Gumbel limit with $\gamma=0$, whereas systematic departures indicate $\gamma>0$.
Let $n_0=\sum_iI(Y_i=0)$, let $X^{(0)}_{n_0:n_0}\ge X^{(0)}_{n_0:n_0-1}\ge\cdots$ be the order statistics of $\{X_i:Y_i=0\}$, fix $k\ge3$, and collect the largest $k$ in $\mathbf X^{(0)}$.
By Theorem \ref{thm:ident}(a) or (b) and the joint limit law of the largest $k$ order statistics \citep[Theorem 2.8.2]{Galambos1978}, there are sequences $a_{n_0}>0$ and $b_{n_0}$ with $(\mathbf X^{(0)}-b_{n_0})/a_{n_0}\dto\mathbf V=(V_1,\ldots,V_k)$, whose density is $G_\gamma(v_k)\prod_{j\le k}g_\gamma(v_j)/G_\gamma(v_j)$ on $v_k\le\cdots\le v_1$ with $g_\gamma=G_\gamma'$.
The normalization constants $a_{n_0}$ and $b_{n_0}$ are unknown, so we use the maximal invariant statistic to location and scale transformations,
\begin{equation}\label{eq:Xstar}
\mathbf X^{*(0)}=\frac{\mathbf X^{(0)}-X^{(0)}_{n_0:n_0-k+1}}{X^{(0)}_{n_0:n_0}-X^{(0)}_{n_0:n_0-k+1}}\ \dto\ \mathbf V^{*}=\frac{\mathbf V-V_k}{V_1-V_k},
\end{equation}
whose density on $\mathcal V=\{v^*:1=v_1^*\ge\cdots\ge v_k^*=0\}$ is
\begin{align}\label{eq:vstar}
f_{\mathbf V^{*}\mid\gamma}(v^{*})&=(k-1)!\int_0^{\infty}t^{k-2}\exp\left\{-\left(1+\tfrac1\gamma\right)\sum_{j=1}^{k}\log\left(1+\gamma v_j^{*}t\right)\right\}dt,\\
\label{eq:vstarnull}
f_{\mathbf V^{*}\mid0}(v^{*})&=\frac{(k-1)!(k-2)!}{\left(\sum_{j}v_j^{*}\right)^{k-1}} ,
\end{align}
where the second expression follows from the first as $\gamma\downarrow0$.
The limit law depends on no parameter other than $\gamma$, so the problem
\begin{equation}\label{eq:hypo}
H_0:\gamma=0\quad\text{against}\quad H_1:\gamma>0
\end{equation}
is one-dimensional with a simple null, and by Theorem \ref{thm:ident}(a) and (b) it is the hypothesis that $q$ is rapidly varying against the hypothesis that $q$ is regularly varying with a finite index.
Following \citet{andrews1994optimal}, we maximize weighted average power with weight $w$ on $(0,\bar\gamma]$ with some upper bound $\bar\gamma>0$, giving
\begin{equation}\label{eq:LR}
R_k=\frac{\int_0^{\bar\gamma}f_{\mathbf V^{*}\mid\gamma}(\mathbf X^{*(0)})w(\gamma)d\gamma}{f_{\mathbf V^{*}\mid0}(\mathbf X^{*(0)})},
\quad
\varphi_k=I\{R_k>\mathrm{cv}(\alpha,k)\}.
\end{equation}
In practice, we set $w$ to be the uniform distribution on $(0,1]$ throughout and use the critical values tabulated in Table \ref{tab:cv} of Appendix \ref{app:sim}.\footnote{The critical values are simulated by drawing $\Gamma_j=\sum_{i\le j}E_i$ with $E_i$ i.i.d.\ standard exponential, setting $V_j=-\log\Gamma_j$, forming $\mathbf V^{*}$, evaluating \eqref{eq:vstar} by Gauss-Legendre quadrature after $t=e^{z}$ centered at $z=\log\{(k-1)/\sum_jv_j^{*}\}$, averaging over a $\gamma$ grid, and taking the $1-\alpha$ quantile.}
The fixed-$k$ null distribution of $R_k$ does not depend on $n$, so rejection is not mechanical in large samples.
The following theorem establishes the asymptotic size control.

\begin{theorem}\label{thm:size}
Suppose $\{(Y_i,X_i)\}_{i=1}^{n}$ is an i.i.d.~sample with $\Prob(Y_i=0)\in(0,1)$ and $F_{X\mid Y=0}\in\mathcal D(G_0)$.
Then for any fixed $k\ge3$ and any $\alpha\in(0,1)$, $\lim_{n\to\infty}\Prob(\varphi_k=1)=\alpha$.
\end{theorem}

The hypothesis $F_{X\mid Y=0}\in\mathcal D(G_0)$ holds under Assumption \ref{ass:main} with part (iii.b).
Furthermore, many laws in the Gumbel class used in practice are symmetric, so both tails are informative.
The right tail statistic uses $X_i\mid Y_i=0$ and the left tail statistic uses $-X_i\mid Y_i=1$, and the two are computed from disjoint subsamples.
The left tail requires an infinite left endpoint, so for nonnegative covariates only the right tail is available.
Write $p_R$ and $p_L$ for the $p$-values of $R_k$ in the two tails, each obtained from the null law of Theorem \ref{thm:size} and each asymptotically uniform on $(0,1)$ under its own null.
Proposition \ref{prop:panel} records the two combination rules used below, one for the two tails within a period and one across periods of a panel, and the panel result that makes the latter necessary.

Consider the panel data model
\[
Y_{it}=I\left(X_{it}+X_{it}^{a\prime}\beta +\rho Y_{i,t-1}+\alpha_i-\varepsilon_{it}\ge0\right),
\]
where $\alpha_i$ is the unobserved individual effect that is constant over $t$, $\rho$ captures the persistence of outcome $Y_{it}$, and $\varepsilon_{it}$ denotes the idiosyncratic error term. 
Let $p_{t,j}$ denote the $p$-value in period $t$ and tail $j\in\{L,R\}$.

\begin{proposition}\label{prop:panel}
\begin{enumerate}[(a).]
\item Suppose $\{(\{Y_{it},X_{it}\}_{t=1}^{T},\alpha_i)\}_{i=1}^{n}$ is i.i.d.\ across $i$ with $T$ fixed, and suppose that for each $t$, $\Prob(Y_{it}=0)\in(0,1)$ and the pair $(F_{X_t},q_t)$, with $q_t(x)=\Prob(Y_{it}=0\mid X_{it}=x)$, satisfies Assumption \ref{ass:main}(ii) and one of the three parts of Assumption \ref{ass:main}(iii).
Then Theorem \ref{thm:ident} applies period by period under the corresponding part of Assumption \ref{ass:main}(iii).
Under part (iii.b), Theorem \ref{thm:size} also applies to the corresponding period-$t$ subsample.
\item Under the sampling conditions of part (a), suppose Assumption \ref{ass:main} and its mirror hold with part (iii.b), so that $\gamma=\gamma^{-}=0$.
Then $p_L$ and $p_R$ are asymptotically independent and $\Prob(\min\{p_L,p_R\}<1-(1-\alpha)^{1/2})\to\alpha$.
\item Suppose the conditions of part (b) hold in every period $t$.
Let
\(
p_t^{S}=1-\left(1-\min\{p_{t,L},p_{t,R}\}\right)^{2}
\)
denote the period-$t$ \v{S}id\'ak-combined $p$-value. Then
\(
\limsup_{n\to\infty}\Prob\left\{\min_{t\le T}p_t^{S}<\frac{\alpha}{T}\right\}\le\alpha.
\)
\end{enumerate}
\end{proposition}

Part (a) extends Remark \ref{rem:panel} to period-specific inference without estimating common coefficients or individual effects, thereby permitting correlated effects and nonstationarity without an incidental parameters problem.
Part (b) combines the two tails within a period, and the two are asymptotically independent because they are computed from disjoint subsamples, the right tail from $\{X_i:Y_i=0\}$ and the left from $\{-X_i:Y_i=1\}$.
The rule is therefore exact rather than conservative, which is why we use \citet{Sidak1967} to combine the two tails within each period.
Part (c) combines the resulting period-level $p$-values by Bonferroni, allowing arbitrary serial dependence.
That matters in practice, since a persistent covariate places the same individuals in the tail in every period, so the period-specific statistics are close to perfectly dependent and a test relying on the independence across time would overreject.

\subsection{Power}\label{sec:power}

Because $k$ is held fixed, the attainable power is bounded, and that bound is the benchmark against which the simulations of Section \ref{sec:mc} and the application of Section \ref{sec:emp} should be read.
This section computes it.
Write $\pi_k(\gamma)$ for the asymptotic power of \eqref{eq:LR} in the limit experiment and $\pi_k^{\mathrm{PO}}(\gamma_1)$ for the power at $\gamma_1$ of the level-$\alpha$ test based on $f_{\mathbf V^{*}\mid\gamma_1}/f_{\mathbf V^{*}\mid0}$.

\begin{proposition}\label{prop:power}
In the limit experiment, let $\varphi$ range over the level-$\alpha$ tests that are invariant to increasing affine transformations of $X$.
\begin{enumerate}[(a).]
\item $\varphi_k$ maximizes $\int_0^{\bar\gamma}\pi(\gamma)w(\gamma)d\gamma$, and $\pi(\gamma_1)\le\pi_k^{\mathrm{PO}}(\gamma_1)$ for any $\gamma_1>0$.
\item For any fixed $k\ge3$ and any $\gamma>0$, $\pi_k(\gamma)<1$.
\end{enumerate}
\end{proposition}

Part (a) shows that the proposed test maximizes the weighted average power, which is less than the power envelope by construction.
Part (b) states that the test is not consistent, and this is a consequence of the framework rather than a shortcoming to be corrected by letting $k$ grow.
Fixing $k$ is what makes \eqref{eq:vstar} an exact limit law free of nuisance parameters, so that Theorem \ref{thm:size} requires no second-order condition, no intermediate sequence, and no bandwidth.
Allowing $k=k_n\to\infty$ would restore consistency but would reintroduce all three, and the resulting size would depend on a rate condition of the form $\sqrt{k_n}A(n/k_n)\to0$ that is not verifiable from data, where $A(n/k_n)$ characterizes the second-order property of the underlying distribution (see Appendix \ref{app:uniform}).
We regard the loss of consistency as the cost of exactness and state it explicitly.
Section \ref{sec:k} accordingly treats $k$ as a design choice before the data are used.

What the test detects is $\gamma$ while the object of interest is $\xi$, and since $\gamma=\lambda\xi/(\lambda+\xi)$,
\[
\gamma-\xi=-\frac{\xi^{2}}{\lambda+\xi},
\quad
\gamma-\xi+\frac{\xi^{2}}{\lambda}=\frac{\xi^{3}}{\lambda(\lambda+\xi)},
\]
so $|\gamma-\xi|\le\xi^{2}/\lambda$ and $|\gamma-\xi+\xi^{2}/\lambda|\le\xi^{3}/\lambda^{2}$ for every $\lambda>0$ and $\xi>0$, and the expansion $\gamma=\xi-\xi^{2}/\lambda+O(\xi^{3})$ holds with an explicit remainder and no restriction on $\xi$ relative to $\lambda$.
Under the baseline model, $\xi$ is also
the extreme-value index of the error. Since $\gamma\mapsto\pi_k(\gamma)$ is continuously differentiable, by differentiation under the integral sign in \eqref{eq:vstar}, the feasible test and the infeasible test that observes the top $k$ order statistics of $\varepsilon_i$ directly have power functions differing by $O(\xi^{2})$ near the null.
The latency of the error therefore costs nothing to first order, and the heaviness of the covariate does not enter the first order approximation.

Table \ref{tab:envelope} reports the power functions.
We find that our test achieves nearly the optimal power at every point, which is a substantial advantage.
For comparison, we compute the \citet{Pickands1975} estimator, which is also invariant to location and scale and therefore a function of $\mathbf V^{*}$.
Figure \ref{fig:power} in Appendix \ref{app:pickands} sets the Pickands critical value by simulation so that the comparison isolates power.
At $k=50$ the Pickands test has power $0.254$ at $\gamma=0.45$ and $0.630$ at $\gamma=1$, compared with $0.810$ and $0.995$ for $\varphi_k$.
If the critical value of the Pickands test is instead taken from the asymptotic standard error $\sqrt3/\{2(\log2)^{2}\sqrt{j_k}\}$ with $j_k=\lfloor k/4\rfloor$, the size is $0.044$ at $k=25$ and $0.040$ at $k=50$, so the deficiency in the limit experiment is in power rather than in size.

\begin{table}[tp]
\centering
\begin{tabular}{l c c c c c c c c}
\hline\hline
 & \multicolumn{8}{c}{$\gamma$} \\
\cline{2-9}
 & 0.05 & 0.10 & 0.20 & 0.25 & 1/3 & 0.50 & 0.75 & 1.00 \\
\hline
$k=10$ & \multicolumn{8}{c}{} \\
\quad test \eqref{eq:LR} & 0.065 & 0.086 & 0.134 & 0.162 & 0.211 & 0.320 & 0.472 & 0.596 \\
\quad power envelope & 0.068 & 0.087 & 0.137 & 0.165 & 0.213 & 0.322 & 0.473 & 0.599 \\
\quad abs.\ difference & 0.003 & 0.001 & 0.003 & 0.003 & 0.002 & 0.002 & 0.001 & 0.003 \\
$k=25$ & \multicolumn{8}{c}{} \\
\quad test \eqref{eq:LR} & 0.085 & 0.129 & 0.247 & 0.315 & 0.425 & 0.622 & 0.822 & 0.921 \\
\quad power envelope & 0.086 & 0.131 & 0.249 & 0.315 & 0.424 & 0.620 & 0.822 & 0.923 \\
\quad abs.\ difference & 0.001 & 0.002 & 0.002 & 0.000 & 0.001 & 0.002 & 0.000 & 0.002 \\
$k=50$ & \multicolumn{8}{c}{} \\
\quad test \eqref{eq:LR} & 0.103 & 0.183 & 0.388 & 0.492 & 0.651 & 0.860 & 0.971 & 0.995 \\
\quad power envelope & 0.106 & 0.187 & 0.390 & 0.492 & 0.651 & 0.860 & 0.972 & 0.996 \\
\quad abs.\ difference & 0.003 & 0.004 & 0.002 & 0.000 & 0.000 & 0.000 & 0.001 & 0.001 \\
\hline
\end{tabular}
\caption{Power of the test \eqref{eq:LR} and the point-optimal power envelope at the $5\%$ level.
The envelope is the rejection probability of the Neyman-Pearson test of $\gamma=0$ against the indicated $\gamma$, which is infeasible because $\gamma$ is unknown.
Critical values are those of Table \ref{tab:cv} in Appendix \ref{app:sim}.
Based on 200,000 simulation draws for each entry.}\label{tab:envelope}
\end{table}

\subsection{Choosing \texorpdfstring{$k$}{k}}\label{sec:k}

The Fisher information of the limit experiment at the null, $\mathcal I_k=\var_0[\partial_\gamma\log f_{\mathbf V^{*}\mid\gamma}(\mathbf V^{*})|_{\gamma=0}]$, equals $4.85$, $18.76$, $43.65$ and $95.35$ at $k=10$, $25$, $50$ and $100$, so $k-\mathcal I_k$ lies between $4.6$ and $6.4$ across that range and the reduction to the maximal invariant costs a bounded amount of information rather than a fixed proportion of it.
Let $\Phi$ denote the standard normal distribution function and $z_p$ its $1-p$ quantile.
The implied approximation $\pi_k(\gamma)\approx\Phi(\sqrt{\mathcal I_k}\gamma-z_\alpha)$ reproduces Table \ref{tab:envelope} closely for moderate $\gamma$, so attaining power $1-\beta$ against a target $\gamma^{\ast}$ requires
\begin{equation}\label{eq:krule}
k\gtrsim\left(\frac{z_\alpha+z_\beta}{\gamma^{\ast}}\right)^{2},
\end{equation}
which gives $k\approx100$ for $\gamma^{\ast}=0.25$ and $k\approx56$ for $\gamma^{\ast}=1/3$ at $\alpha=0.05$ and power $0.8$.
The detectable magnitude is therefore of order $k^{-1/2}$ and does not improve with $n$, which is the content of Proposition \ref{prop:power}(b).

Rule \eqref{eq:krule} is a lower bound driven by power, while the upper bound is driven by the extreme value approximation, which deteriorates as $k$ grows relative to $n_0$.
Table \ref{tab:cross_sectional} shows the resulting size distortion at $k=70$ with $n=1000$.
We therefore recommend reporting the test at several values of $k$ rather than at one, as we do in Table \ref{tab:innovation}, and reading the results as a robustness check.
Agreement across $k$ is evidence that the extreme value approximation is adequate over the range used, and disagreement is a warning that it is not.

\subsection{Uniform Size}\label{sec:uniform}

Theorem \ref{thm:size} is pointwise in $F$.
No uniform statement is available over the unrestricted null, and uniformity is restored on a second-order restricted subclass.
Theorem \ref{thm:uniform} in Appendix \ref{app:uniform} states both, and the construction behind the first is a law that is exactly Pareto below a threshold $M$ and has a constant reciprocal hazard beyond it.
Such a law lies in the null for every $M$, yet its largest $k$ order statistics are exactly Pareto with probability approaching one as $M$ grows at fixed $n$, so the test rejects it at the power it would have against a genuine alternative.
The restriction that repairs this is the second-order condition of extreme value theory imposed with constants common to the class, which the construction violates because its remainder is negligible only beyond $M$.

\section{Monte Carlo Simulations} \label{sec:mc}

We assess finite sample size and power of the proposed test \eqref{eq:LR} in both cross-sectional and panel designs in Sections \ref{sec:simu:cross-sectional} and \ref{sec:simu:panel}, respectively.

\subsection{Cross-Sectional Model}\label{sec:simu:cross-sectional}
We start with the cross-sectional case and generate data from the following model
\[
    Y_i = I\left(X_{i} + X_{1,i}^a + X_{2,i}^a - \varepsilon_i \ge 0\right),
\]
where the error term \(\varepsilon_i\) is drawn from four distributions: standard normal, logistic, and two Student-\(t\) distributions with different degrees of freedom. The normal and logistic distributions represent the null hypothesis, while the Student-\(t\) distributions serve as heavy-tailed alternatives. The dominating heavy-tailed covariate \(X_i\) is drawn from a Student-$t(2)$ distribution. The other two components, \(X_{1,i}^a\) and \(X_{2,i}^a\), are thin-tailed and drawn from a standard normal distribution and a binary distribution on $\pm1/2$ with equal probability, respectively. 
We consider sample sizes of \(n \in \{1000,2000,5000 \} \).

We apply the test using only the dominating covariate $X_i$, so no model estimation is required, which in turn reduces potential misspecification issues. We report rejection rates for the left and right tails of $\varepsilon_i$ separately in columns labeled Left and Right in Table \ref{tab:cross_sectional}, and for the joint null that both tails are thin using \v{S}id\'ak critical values in columns labeled Both. Results are shown for various numbers of tail observations $k\in\{10,25,50,70\}$. 

We summarize the findings in Table \ref{tab:cross_sectional} as follows. First, size is well controlled under the null for both normal and logistic errors, lying in the Gumbel domain. The test yields slightly more conservative rejection rates under normal than under logistic errors, since the normal tail decays as $\exp(-cx^2)$ and the logistic as $\exp(-cx)$. Moreover, mild size distortions arise when $k$ is large relative to $n$, such as $k=70$ with relatively small $n$, as too many mid-sample order statistics then enter the tail approximation and the extreme value approximation becomes rough.

Second, the test exhibits strong power against heavy-tailed alternatives, even for small $k$ such as $k=10$. Power increases with $k$ and with tail heaviness, with heavier tails such as $t(1)$ being detected more easily than $t(2)$. Joint testing of both tails yields higher power, as expected. Note that power increases mainly with $k$ rather than $n$, since larger $n$ improves the extreme-value approximation but does not directly raise power. 

\begin{table}[tp]
\centering
\resizebox{\textwidth}{!}{
\begin{tabular}{lcrrrrrrrrrrrr}
\hline\hline
Dist.\ of $\varepsilon$ & $n$ & \multicolumn{3}{c}{$k=10$} & \multicolumn{3}{c}{$k=25$} & \multicolumn{3}{c}{$k=50$} & \multicolumn{3}{c}{$k=70$} \\
 &  & Left & Right & Both & Left & Right & Both & Left & Right & Both & Left & Right & Both \\
\hline
Normal  & 1000 & 0.03 & 0.03 & 0.02 & 0.02 & 0.02 & 0.02 & 0.01 & 0.01 & 0.01 & 0.01 & 0.00 & 0.00 \\
 & 2000 & 0.03 & 0.03 & 0.03 & 0.02 & 0.03 & 0.02 & 0.01 & 0.02 & 0.01 & 0.01 & 0.01 & 0.01 \\
 & 5000 & 0.02 & 0.02 & 0.03 & 0.03 & 0.02 & 0.02 & 0.02 & 0.02 & 0.01 & 0.01 & 0.01 & 0.01 \\
Logistic & 1000 & 0.04 & 0.04 & 0.05 & 0.04 & 0.04 & 0.05 & 0.04 & 0.04 & 0.03 & 0.02 & 0.03 & 0.03 \\
 & 2000 & 0.05 & 0.06 & 0.05 & 0.05 & 0.05 & 0.05 & 0.05 & 0.05 & 0.06 & 0.06 & 0.06 & 0.06 \\
 & 5000 & 0.05 & 0.05 & 0.05 & 0.05 & 0.06 & 0.06 & 0.06 & 0.05 & 0.06 & 0.06 & 0.06 & 0.06 \\ \hline
$t(2)$ & 1000 & 0.09 & 0.10 & 0.12 & 0.14 & 0.13 & 0.18 & 0.14 & 0.15 & 0.17 & 0.21 & 0.21 & 0.21 \\
 & 2000 & 0.12 & 0.11 & 0.14 & 0.17 & 0.17 & 0.24 & 0.21 & 0.21 & 0.30 & 0.30 & 0.29 & 0.31 \\
 & 5000 & 0.13 & 0.15 & 0.19 & 0.21 & 0.22 & 0.30 & 0.32 & 0.31 & 0.47 & 0.47 & 0.47 & 0.47 \\
$t(1)$  & 1000 & 0.17 & 0.17 & 0.23 & 0.31 & 0.29 & 0.41 & 0.37 & 0.38 & 0.51 & 0.51 & 0.51 & 0.51 \\
 & 2000 & 0.18 & 0.18 & 0.23 & 0.33 & 0.34 & 0.46 & 0.46 & 0.48 & 0.71 & 0.71 & 0.71 & 0.71 \\
 & 5000 & 0.18 & 0.18 & 0.24 & 0.37 & 0.35 & 0.50 & 0.56 & 0.56 & 0.73 & 0.65 & 0.66 & 0.83 \\
\hline
\end{tabular}}
\caption{Simulation results for cross-sectional data. Entries are the rejection rates of our proposed test \eqref{eq:LR}. The normal and logistic distributions correspond to the null hypothesis, and the Student-$t$ distributions correspond to the alternative. Significance level is $5\%$. Based on 2,000 simulation draws.}
\label{tab:cross_sectional}
\end{table}

We finally compare the test with a specification test of the parametric model itself.
The natural comparison is the information matrix test of \citet{White1982}, computed in its outer product form for a logit fitted to $(X_i,X_{1,i}^a,X_{2,i}^a)$.
Its null is that the logit is correctly specified, so the logistic design is its null and every other design is an alternative for it.
Table \ref{tab:im} reports its rejection rates alongside those of $\varphi_k$ at $k=50$, with the critical value simulated under the logistic design so that the comparison is not affected by the accuracy of the $\chi^2$ approximation.

The two tests differ in what they treat as the null, and the normal row shows what that costs.
Normal and logistic errors are both thin-tailed, both lie in the Gumbel class, and both are null for \eqref{eq:hypo}, so $\varphi_k$ should reject neither.
The information matrix test rejects the normal design at $0.20$, $0.28$ and $0.56$, and the rate rises with $n$.
Size correction does not remove this, because it is not a size distortion: the normal design is a genuine alternative for a fitted logit, and the test achieves power.
However, a rejection indicates that the error is not logistic, which says nothing about whether it is thin-tailed.
If we fit a probit instead of the logit, the results would swap and the logistic design would be the one rejected.

Regarding power, our test dominates at $n=1000$, attaining $0.17$ and $0.51$ against $0.06$ and $0.16$, and the two are close at $n=2000$.
When the sample size is large at $n=5000$, the information matrix test has larger power.
This pattern is consistent with the argument in Remark \ref{rem:moments}.
The information matrix test draws its power from the mid-sample and so needs $n$ to grow, whereas $\varphi_k$ reads a fixed number of extremes and has power already at $n=1000$ but improves little thereafter with a fixed $k$.

\begin{table}[tp]
\centering
\begin{tabular}{lcccc}
\hline\hline
Dist.\ of $\varepsilon$ & $n$ & \multicolumn{2}{c}{information matrix test} & test \eqref{eq:LR} \\
\cline{3-4}
 &  & asymptotic cv & simulated cv & $k=50$, both tails \\
\hline
Normal   & 1000 & 0.28 & 0.20 & 0.01 \\
         & 2000 & 0.30 & 0.28 & 0.01 \\
         & 5000 & 0.49 & 0.56 & 0.01 \\
Logistic & 1000 & 0.10 & 0.05 & 0.03 \\
         & 2000 & 0.06 & 0.06 & 0.06 \\
         & 5000 & 0.03 & 0.04 & 0.06 \\
\hline
$t(2)$   & 1000 & 0.12 & 0.06 & 0.17 \\
         & 2000 & 0.23 & 0.19 & 0.30 \\
         & 5000 & 0.79 & 0.84 & 0.47 \\
$t(1)$   & 1000 & 0.32 & 0.16 & 0.51 \\
         & 2000 & 0.71 & 0.67 & 0.71 \\
         & 5000 & 1.00 & 1.00 & 0.73 \\
\hline
\end{tabular}
\caption{Rejection rates at the $5\%$ level for the information matrix test of \citet{White1982} applied to a fitted logit and for test \eqref{eq:LR}. The information matrix critical values are asymptotic $\chi^2_6$ or simulated under logistic errors (based on 1,500 draws). Only logistic errors satisfy the information matrix null, whereas normal and logistic errors satisfy \eqref{eq:hypo}. The test \eqref{eq:LR} entries are from Table \ref{tab:cross_sectional}.}
\label{tab:im}
\end{table}

\subsection{Panel Data Models}\label{sec:simu:panel}
We examine both static and dynamic panel models with unobserved individual heterogeneity. In the static case, we generate data from
\begin{equation}\label{eq:non-dynamic}
    Y_{it} = I(X_{it}+X_{1,it}^a+X_{2,it}^a + \alpha_i- \varepsilon_{it}\ge 0),
\end{equation}
where $\{X_{it},X_{1,it}^a,X_{2,it}^a,\varepsilon_{it}\}$ are i.i.d.\ across $i$ and $t$ with the same marginal distributions as in the cross-sectional design. The individual effect $\alpha_i$ is drawn from the standard uniform distribution $U(0,1)$. We consider sample sizes $n \in \{1000,2000,5000\}$ and $T=2$. In the dynamic case, the data generating process is
\[
    Y_{it} = I(X_{it}+Y_{i,t-1}+X_{1,it}^a+X_{2,it}^a + \alpha_i- \varepsilon_{it}\ge 0),
\]
with the initial period $Y_{i0}$ generated from \eqref{eq:non-dynamic}. All other specifications match the static panel design.

We implement the test period by period, combining the two tails within each period by the \v{S}id\'ak rule and the two periods by the Bonferroni rule of Proposition \ref{prop:panel}(c). The results, reported in Tables \ref{tab:panel_nondynamic} and \ref{tab:panel_dynamic} of Appendix \ref{app:sim}, show patterns similar to the cross-sectional case. The test maintains good size control, and power increases with $k$ and tail heaviness. As expected, the Bonferroni correction is slightly conservative, but the test remains powerful. Overall, size and power behave similarly across cross-sectional, static panel, and dynamic panel settings, indicating that our test is well-suited to a wide range of data structures.
 
\section{Empirical Example}\label{sec:emp}
We apply the proposed diagnostic test to firms' binary decision to innovate.
In the Schumpeterian growth literature, innovation and firm size are closely related, because innovation generates rents and market leadership, and firm scale affects the returns to R\&D and the ability to appropriate those rents \citep{cohen2010fifty,aghion2014we}.
Firm size is well known to display a heavy-tailed distribution: see for example \citet{Axtell2001}, \citet{gabaix2009power}, and \citet{diGiovanniLevchenko2012}.
We use a single dominating heavy-tailed covariate $X_i$, and the test infers the error tail from the conditional extremes of $X_i$ without estimating the binary choice model, so it is simple to implement and robust to potential model misspecification.

Following the cross-sectional setup with multiple covariates in \eqref{eq:multi}, the binary outcome $Y_i$ is an indicator for whether firm $i$ files at least one patent in a given year, so that $Y_i=1$ if the firm patents and $Y_i=0$ otherwise.
Firm size can be measured in several ways, including employment, assets, output, and value added, and here we use value added as the measure of scale, which is the resource base for R\&D and is closest to the notion of scale in the Schumpeterian literature.
The dominating regressor $X_i$ is therefore the firm's lagged value added, which captures pre-existing operating scale.
Ownership indicators, policy variables, and industry or region fixed effects enter as auxiliary regressors $X_i^a$.

We merge the National Tax Survey Database (NTSD) from China with patent filing data from China's State Intellectual Property Office (SIPO).
The NTSD is jointly administered by the Chinese State Administration of Taxation and the Ministry of Finance, and is a stratified random sample of taxpayers drawn each year and broadly regarded as the most reliable source of firm-level tax and performance information in China.
It contains detailed tax data and firm characteristics, including the value-added measure used in our analysis.
The dataset has been widely used in prior work such as \citet{chenTaxPolicyLumpy2023} and \citet{liuHowTaxIncentives2019}.
Since all patent applications in China must be submitted to SIPO, the SIPO records provide comprehensive information on patent filings and grants, including the filing date, applicant name and address, firm name, and patent type.
Several previous studies have used the SIPO data, such as \citet{liuIntermediateInputImports2016} and \citet{liuTradePolicyUncertainty2020}.

\begin{figure}[htbp]
    \centering
    \setlength{\tabcolsep}{0pt}
    \begin{tabular}{cc}
    \includegraphics[width=0.44\textwidth]{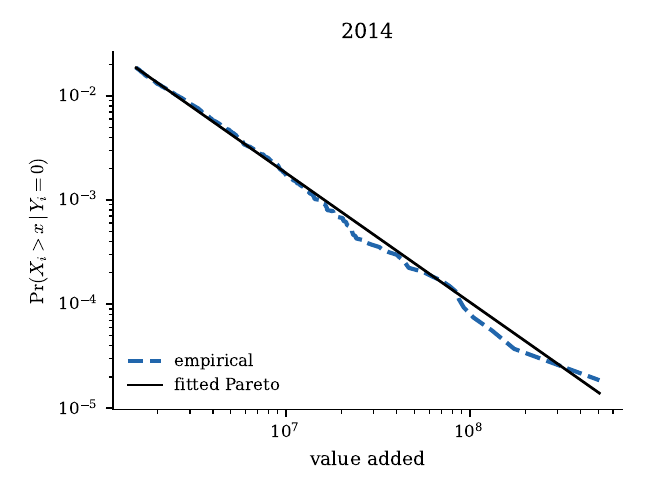} &
    \includegraphics[width=0.44\textwidth]{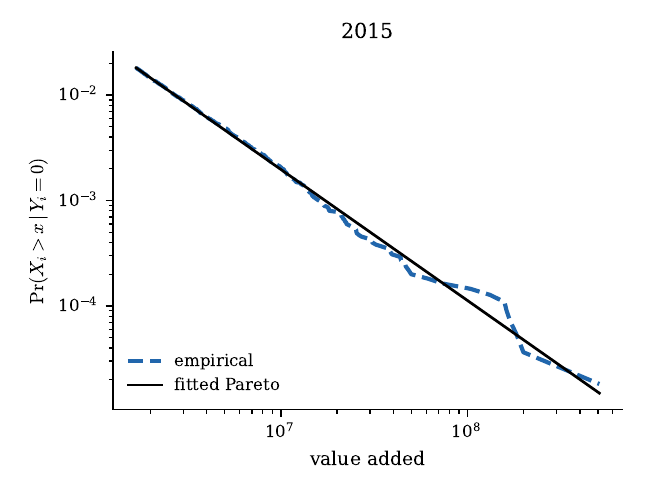}
    \end{tabular}
    \caption{Pareto fit of lagged firm value added among non-patenting firms ($X_i\mid Y_i=0$): empirical (blue dashed) and fitted Pareto (black solid), largest $1{,}000$ observations.}
    \label{fig:valueadd}
\end{figure}

We construct our sample from the merged data using observations for which the patenting indicator and a positive lagged value added are both available.
Since the auxiliary regressors are finite-support discrete or bounded, the conditional moment condition \eqref{eq:condmom} holds.
For the structural interpretation of the null, we maintain the tail-stability condition \eqref{eq:tailstable}.
We assess its empirical plausibility by comparing the distributions of the auxiliary covariates among firms in the upper $5\%$, $2\%$, and $1\%$ of the lagged-value-added distribution.
The distributions of the ownership, policy, industry, and region controls change only modestly across adjacent thresholds and exhibit no systematic deterioration farther into the tail.
This evidence supports the tail-stability condition in the present application.
Figure \ref{fig:valueadd} plots the conditional survival function of lagged value added among non-patenting firms, $X_i\mid Y_i=0$, over the largest $1{,}000$ observations, against a fitted Pareto law.
The two curves track each other closely across two orders of magnitude, so the domain of attraction assumption is plausibly satisfied and the conditional distribution exhibits a heavy upper tail.

Following Section \ref{sec:k}, we compute the likelihood-ratio statistic \eqref{eq:LR} at several values of $k$ and read the results as a robustness check.
Rule \eqref{eq:krule} gives $k\approx100$ for a target of $\gamma^{\ast}=0.25$, and $k=100$ remains a small fraction of $n_0$ here, so we take $k\in\{25,50,70,100\}$.
Table \ref{tab:innovation} reports the resulting $p$-values for two annual cross sections.
The conclusion is robust to the choice of $k$: the $p$-value is below $0.01$ for every $k\ge50$ in both years, and the only value above $0.05$ is $0.17$ at the smallest $k$ in 2014, where \eqref{eq:krule} predicts the least power.
We therefore reject the thin-tail null for the latent error.
Economically, this suggests that the decision to innovate is influenced by rare but sizable unobserved shocks, such as technological opportunities, credit conditions, or policy changes, so that conventional probit or logit specifications may understate the probability of innovation in the relevant tail.

\begin{table}[htbp]
\centering
\begin{tabular}{lrrrr}
\hline\hline
\multicolumn{5}{c}{$H_0:$ Right tail of $\varepsilon$ is thin} \\
Sample & $k=25$ & 50 & 70 & 100 \\  \hline
2014  & 0.17 & $<$0.01 & $<$0.01 & $<$0.01 \\
2015  & $<$0.01 & $<$0.01 & $<$0.01 & $<$0.01 \\
\hline
\end{tabular}
\caption{The $p$-values of our test \eqref{eq:LR} in the firm innovation dataset. Here $Y_i$ indicates whether the firm files at least one patent, and $X_i$ is the firm's lagged value added.}
\label{tab:innovation}
\end{table}

Two features of this design deserve comment.
First, the dominance requirement of Section \ref{sec:multi} is met, since lagged value added is the only regressor with unbounded support and the remaining covariates are finite-support discrete or bounded.
Second, conditioning on $Y_i=0$ tilts the tail, which is what Theorem \ref{thm:ident}(a) implies.
The Hill estimates for 2015 are $0.854$ for lagged value added of all firms and $0.786$ among non-patenting firms, a difference of $-0.068$ with bootstrap standard error $0.020$ and $95\%$ interval $[-0.109,-0.035]$, computed at the top $1{,}000$ order statistics.
The corresponding difference in 2014 is $-0.058$ with interval $[-0.108,-0.011]$, and the sign is the same at the top $500$ and the top $2{,}000$.
The conditional tail is therefore thinner than the unconditional tail, which is the ordering $\gamma<\lambda$ of Theorem \ref{thm:ident}(a), and the rejection reflects the tail of the latent error rather than the configuration of Theorem \ref{thm:ident}(c) in which the covariate fails to shift the choice probability at a polynomial rate.

\section{Conclusion}\label{sec:conclusion}

This paper develops a simple yet powerful diagnostic test for tail behavior of binary choice models. Using normalized spacings of extreme covariate order statistics conditional on the binary outcome, the test distinguishes the thin-tailed behavior implied by conventional logit and probit specifications from heavy-tailed alternatives without requiring model estimation. The likelihood-ratio statistic is invariant to location and scale and extends naturally to models with multiple covariates and to panel models with individual effects and dynamics.

Monte Carlo evidence shows good size control and high power, even in modest samples. An application to firm innovation illustrates how the diagnostic can assess the adequacy of conventional link functions and the potential importance of rare but large shocks.

{\setstretch{0.85}\setlength{\bibsep}{0pt}\footnotesize\raggedright
\bibliographystyle{apalike}
\bibliography{reference}}

\clearpage
\appendix
\titleformat{\section}[block]{\large\bfseries}{Appendix \thesection:}{0.6em}{}
\renewcommand{\thesection}{\Alph{section}}
\setcounter{page}{1}
\renewcommand{\thepage}{A\arabic{page}}
\markboth{\normalfont Ji, Liu, Wang, and Xing}{\normalfont Supplement to ``A Diagnostic Test for Binary Choice Models''}

\begin{center}
{\large\bfseries Supplementary Material for\\[2pt]
``A Simple and Powerful Diagnostic Test for Binary Choice Models''}

\medskip

{\setstretch{1.0}\normalsize Ting Ji \quad Laura Liu \quad Yulong Wang \quad Jiahe Xing}
\end{center}

\bigskip

{\setstretch{1.0}\noindent This supplement collects material referenced in the paper.
Appendix \ref{app:uniform} states and discusses the uniform size result summarized in Section \ref{sec:uniform}.
Appendix \ref{app:proofs} proves every result in the paper and in this supplement.
Appendix \ref{app:pickands} compares the asymptotic power of our test with that of a test based on the Pickands estimator.
Appendix \ref{app:sim} collects the critical values of the test and the panel simulation results.
Appendix \ref{app:plugin} develops the plug-in version of the test for the case in which no single covariate is known to dominate the tail.
Equation, table, figure, and result numbers continue those of the paper.}

\bigskip
\section{Uniform Size}\label{app:uniform}

Theorem \ref{thm:size} establishes the asymptotic size of $\varphi_k$ at each fixed null law, and Section \ref{sec:uniform} summarizes what happens when the statement is required to hold uniformly.
This appendix gives the two results in full, together with the class over which uniformity is recovered and a numerical illustration of the failure.
Theorem \ref{thm:size} is pointwise in $F$.
Theorem \ref{thm:uniform} shows that no uniform statement is available over the unrestricted null, and that uniformity is restored on a second-order restricted subclass.

Let $\mathcal F_{0}$ denote the set of joint laws of $(Y_i,X_i)$ satisfying Assumption \ref{ass:main} with part (iii.b), so that $\gamma=0$ throughout $\mathcal F_{0}$.
For $F\in\mathcal F_{0}$ write $U_F=(1/\{1-F_{X\mid Y=0}\})^{\leftarrow}$ and, for $\rho<0$, $A_{0}>0$, $t_{0}\ge1$ and $p_{0}\in(0,1)$, define the second-order class
\begin{multline}\label{eq:uniclass}
\mathcal F_{0}(\rho,A_{0},t_{0},p_{0})=\left\{F\in\mathcal F_{0}:\ \Prob(Y_i=0)\ge p_{0}\ \text{ and }\ \exists a_F>0\ \text{ with}\right.\\
\left.\sup_{x>1/t}\frac{\left|U_F(tx)-U_F(t)-a_F(t)\log x\right|}{a_F(t)(1+|\log x|)}\le A_{0}t^{\rho}\ \text{ for all }t\ge t_{0}\right\}.
\end{multline}
The bound in \eqref{eq:uniclass} is the second-order condition of extreme value theory imposed uniformly over $\mathcal F_{0}$. See \citet{Drees1998} for the corresponding uniform inequalities for the tail quantile process.

Both statements use only that the null law of $R_k$ has no atoms, which Lemma \ref{lem:analytic} in the appendix establishes.

\begin{theorem}\label{thm:uniform}
\begin{enumerate}[(a).]
\item For any fixed $k\ge3$,
\[
\liminf_{n\to\infty}\ \sup_{F\in\mathcal F_{0}}\Prob_F(\varphi_k=1)\ \ge\ \sup_{\gamma>0}\pi_k(\gamma)\ >\ \alpha .
\]
\item For any fixed $k\ge3$ and any $\rho<0$, $A_{0}>0$, $t_{0}\ge1$ and $p_{0}\in(0,1)$,
\[
\lim_{n\to\infty}\ \sup_{F\in\mathcal F_{0}(\rho,A_{0},t_{0},p_{0})}\left|\Prob_F(\varphi_k=1)-\alpha\right|=0 .
\]
\end{enumerate}
\end{theorem}

Table \ref{tab:uniform} illustrates part (a).
With $\gamma_{0}=1/2$, $n_0=2000$ and $k=25$, the test rejects a law that satisfies Assumption \ref{ass:main}(iii.b), and therefore has $\gamma=0$, between $61$ and $66$ percent of the time at nominal $5$ percent, matching the reference value $\pi_k(1/2)=0.629$.
The distortion is present already at $M=10^{2}$ because the largest $k$ observations of a sample of this size lie far below $M$.

\begin{table}[tp]
\centering
\small
\begin{tabular}{lcccccc}
\toprule
$M$ & $10^{2}$ & $10^{3}$ & $10^{4}$ & $10^{6}$ & $10^{10}$ & $\pi_k(1/2)$\\
\midrule
rejection rate & 0.634 & 0.612 & 0.619 & 0.616 & 0.620 & 0.618\\
\bottomrule
\end{tabular}
\caption{Rejection rate of $\varphi_k$ at nominal $\alpha=0.05$ under the laws $F_M$ constructed in the proof of Theorem \ref{thm:uniform}(a), all of which satisfy Assumption \ref{ass:main} with part (iii.b) and hence have $\gamma=0$. Here $\gamma_{0}=1/2$, $n_0=2000$, $k=25$, based on 2,000 replications.}
\label{tab:uniform}
\end{table} 

Part (a) is proved by a switching construction.
A law that is exactly Pareto with index $1/\gamma_{0}$ below a threshold $M$, and is exponential above $M$, lies in $\mathcal F_{0}$ for every $M$.
Its top $k$ order statistics are nevertheless exactly Pareto with probability approaching one as $M\to\infty$ at fixed $n$.
The exponential scale is chosen so that the conditional density is continuous at $M$.
The rejection probability is then $\pi_k(\gamma_{0})$, which exceeds $\alpha$.
The restriction in \eqref{eq:uniclass} rules this out by forcing the second-order remainder to vanish at a rate common to the class, and this is the sense in which a uniformity restriction is unavoidable rather than technical.


\section{Proofs}\label{app:proofs}

Two classical results are used repeatedly.
Karamata's theorem states that if $h\in\RV_{-\varrho}$ is positive and locally bounded on $[z_0,\infty)$ and $\varrho>1$, then $\int_x^{\infty}h(u)du<\infty$ for large $x$ and $\int_x^{\infty}h(u)du\sim xh(x)/(\varrho-1)$, while if $\varrho<1$ then $\int_{z_0}^{x}h(u)du\sim xh(x)/(1-\varrho)$ \citep[Proposition 1.5.8 and Theorem 1.5.11]{Bingham1987}.
Von Mises' condition states that if $F$ has right endpoint $+\infty$ and a positive differentiable derivative $f$ on $(z_0,\infty)$, and if $\psi=(1-F)/f$ satisfies $\psi'(x)\to0$, then $F\in\mathcal D(G_0)$ with auxiliary function $\psi$ \citep[Theorem 1.1.8 and Remark 1.1.10]{deHaan07}.

We also use three standard properties of $\Gamma$ variation.
If $h$ is $\Gamma$ varying with auxiliary function $a$, then $a(t)=o(t)$ and is determined up to asymptotic equivalence.
Moreover, $h$ is integrable at infinity, $\int_t^{\infty}h(u)du\sim a(t)h(t)$, and $\int_t^{\infty}h$ is itself $\Gamma$ varying with the same auxiliary function.
Finally, a distribution $F$ with infinite right endpoint lies in $\mathcal D(G_0)$ if and only if $1-F$ is $\Gamma$ varying \citep[Section 1.2]{deHaan07}.

We also use the following consequence of the Karamata representation theorem \citep[Theorem 1.3.1]{Bingham1987}.
If $h>0$ has $h'/h$ continuous on $(z_0,\infty)$ and $xh'(x)/h(x)\to-\varrho$, then integrating $h'/h$ gives
\[
h(x)=h(z_0)x^{-\varrho}z_0^{\varrho}\exp\left\{-\int_{z_0}^{x}\eta(u)u^{-1}du\right\}
\]
with $\eta(u)\to0$.
The exponential factor is slowly varying by that theorem, so $h\in\RV_{-\varrho}$.

Throughout, $g=qf_X$, which is positive and continuous on $(z_0,\infty)$ under Assumption \ref{ass:main}.
Bayes' rule gives
\begin{equation}\label{eq:tailrep}
f_{X\mid Y=0}=\frac{g}{\Prob(Y_i=0)},\quad 1-F_{X\mid Y=0}(x)=\frac{1}{\Prob(Y_i=0)}\int_x^{\infty}g(u)du .
\end{equation}
Since $q>0$ and $F_X$ has right endpoint $+\infty$, so does $F_{X\mid Y=0}$.
The two lemmas below are proved first, and the remaining proofs are given in the order in which the results are stated.

\begin{lemma}\label{lem:rapid}
Let $\phi>0$ be differentiable on $(z_1,\infty)$ with $\phi'(x)\to0$.
Then $\phi(x)/x\to0$.
If in addition $\phi=-q/q'$ for some positive differentiable $q$ with $q'<0$ on $(z_1,\infty)$, then $q$ is rapidly varying, and for every $M>0$ there are $z_M$ and $C_M$ with $q(x)\le C_Mx^{-M}$ for $x\ge z_M$.
\end{lemma}

\begin{proof}
Fix $\epsilon>0$ and choose $z_\epsilon>z_1$ with $|\phi'(u)|\le\epsilon$ for $u\ge z_\epsilon$.
For $x>z_\epsilon$ the mean value theorem gives $\phi(x)-\phi(z_\epsilon)=\phi'(\zeta)(x-z_\epsilon)$ for some $\zeta\in(z_\epsilon,x)$, so
\[
0<\frac{\phi(x)}{x}\le\frac{\phi(z_\epsilon)}{x}+\epsilon .
\]
Hence $\limsup_{x\to\infty}\phi(x)/x\le\epsilon$.
Since $\epsilon$ was arbitrary, $\phi(x)/x\to0$.

For the second claim, $\phi$ is continuous and positive, so $u\mapsto\phi(u)^{-1}$ is continuous and $\log q$ is differentiable with $(\log q)'=q'/q=-1/\phi$.
The fundamental theorem of calculus gives $\log\{q(y)/q(x)\}=-\int_x^{y}\phi(u)^{-1}du$ for $z_1<x<y$.
Fix $M>0$ and use $\phi(u)/u\to0$ to choose $z_M>z_1$ with $\phi(u)\le u/M$ for $u\ge z_M$.
Then $\int_x^{y}\phi(u)^{-1}du\ge M\log(y/x)$ for $y>x\ge z_M$, so taking $x=z_M$ yields $q(y)\le q(z_M)z_M^{M}y^{-M}$.
This is the stated bound with $C_M=q(z_M)z_M^{M}$.
Taking $y=tx$ with $t>1$ gives $q(tx)/q(x)\le t^{-M}$ for $x\ge z_M$, hence $\limsup_{x\to\infty}q(tx)/q(x)\le t^{-M}$ for every $M>0$.
Letting $M\to\infty$ gives rapid variation.
\end{proof}

Parts (iii.a) and (iii.b) are mutually exclusive because a $\Gamma$ varying function is rapidly varying and so lies in no $\RV_{\theta}$ with $\theta$ finite.
Case (iii.c) is $q\in\RV_{0}$, which is exclusive of both, because (iii.a) gives $q\in\RV_{-1/\xi}$ and (iii.b) gives rapid variation of $q$.

\begin{proofof}{Proposition \ref{prop:consequences}}
Write $T_i(c)=\{Y_i-I(X_i>0)\}f_X(X_i)^{-1}I(|X_i|\le c)$, so that $\widehat\theta_n=n^{-1}\sum_iT_i(c_n)$ is an average of i.i.d.\ terms for each fixed $c_n$.
Conditioning on $X_i$ and using $\E[Y_i\mid X_i=u]=p(u)$,
\[
\E[T_i(c)]=\int_{-c}^{c}\frac{p(u)-I(u>0)}{f_X(u)}f_X(u)du=\int_{-c}^{c}\{p(u)-I(u>0)\}du ,
\]
so that $\lim_{c\to\infty}\E[T_i(c)]=\int_{-\infty}^{\infty}\{p(u)-I(u>0)\}du$, which is the identified functional of \citet{Lewbel1997}.
Under \eqref{eq:lewbel}, $p(u)=F_\varepsilon(\theta_0+u)$, so the substitution $t=\theta_0+u$ gives $\int_{-\infty}^{\infty}\{F_\varepsilon(t)-I(t>\theta_0)\}dt=\theta_0-\E[\varepsilon_i]$, which equals $\theta_0$ by the normalization in \eqref{eq:lewbel}.
The functional therefore equals the intercept and $\lim_{c\to\infty}\E[T_i(c)]=\theta_0$.

Subtracting and using $I(u>0)=0$ on $\{u<-c\}$ and $p(u)-1=-q(u)$ on $\{u>c\}$,
\begin{align*}
\E[T_i(c)]-\theta_0=& -\int_{-\infty}^{-c}\{p(u)-I(u>0)\}du-\int_{c}^{\infty}\{p(u)-I(u>0)\}du\\
=&\int_{c}^{\infty}q(u)du-\int_{-\infty}^{-c}p(u)du=b(c) .
\end{align*}
Since $\{Y_i-I(X_i>0)\}^{2}=|Y_i-I(X_i>0)|$ takes values in $\{0,1\}$ and has conditional mean $p(u)I(u\le0)+q(u)I(u>0)$ given $X_i=u$,
\begin{equation}\label{eq:ktvar}
\E[T_i(c)^{2}]=\int_{-c}^{c}\frac{p(u)I(u\le0)+q(u)I(u>0)}{f_X(u)}du=v(c).
\end{equation}
This quantity is finite for every finite $c$ because the numerator is bounded by one and $1/f_X$ is locally integrable.
Therefore, $\var\{T_i(c)\}=v(c)-\E[T_i(c)]^{2}$.
Because $\widehat\theta_n$ is an average of i.i.d.\ terms,
\begin{equation}\label{eq:ktexp}
\E\left\{\left(\widehat\theta_n-\theta_0\right)^{2}\right\}
=\frac{\var\{T_i(c_n)\}}{n}+b(c_n)^{2},
\end{equation}
which is an identity and so delivers both bounds.
By hypothesis the left tail of \eqref{eq:ktvar} contributes a term of the same order as the right, and hence we treat the right tail.

For $u>0$ the integrand of \eqref{eq:ktvar} is $q/f_X\in\RV_{\theta}$ with $\theta=1/\lambda+1-1/\xi$.
Karamata's theorem is applied to it with index $\theta=\nu-1$ in the notation above, where $\nu:=\theta+1=2+1/\lambda-1/\xi$.
If $\theta<-1$, equivalently $\nu<0$, then Karamata's theorem gives $\int^{\infty}q/f_X<\infty$ and $v(c)=O(1)$.
If $\theta>-1$, that is $\nu>0$, Karamata's theorem gives $\int^{c}q/f_X\sim c\{q(c)/f_X(c)\}/\nu\in\RV_{\nu}$, and hence $v(c)=c^{\nu+o(1)}$ because a regularly varying function of index $\nu$ equals $c^{\nu}$ times one of index zero, which is $c^{o(1)}$.
If $\nu=0$ the integral is slowly varying and $v(c)=c^{o(1)}$.

The bias is a difference of two tail integrals of the same order, and the two may cancel.
If $\varepsilon_i$ is symmetric and $\theta_0=0$ then $b\equiv0$.
Only a bound is therefore available.
Since $\xi<1$, both $\int_c^{\infty}q$ and $\int_{-\infty}^{-c}p$ are finite, and Karamata's theorem gives $\int_c^{\infty}q\sim cq(c)/(1/\xi-1)=c^{1-1/\xi+o(1)}$, with the same order for the left tail by hypothesis.
The triangle inequality then gives
\begin{equation}\label{eq:ktbias}
|b(c)|\le c^{1-1/\xi+o(1)} .
\end{equation}
The lower bound in \eqref{eq:rate} comes from the variance instead, which is why no matching lower bound on $|b|$ is needed.

Let $c_n=n^{\gamma+o(1)}$ and consider the two regimes in turn.

If $\nu>0$ then $v(c_n)=n^{\gamma\nu+o(1)}\to\infty$, while $\E[T_i(c_n)]\to\theta_0$ is bounded, so $\var\{T_i(c_n)\}\sim v(c_n)$ and the first term of \eqref{eq:ktexp} is $n^{\gamma\nu-1+o(1)}$.
Substituting $\gamma/\xi=1-\gamma/\lambda$, which is the definition of $\gamma$ rearranged, gives $\gamma\nu=2\gamma+2\gamma/\lambda-1$, so
\[
\gamma\nu-1=-2\left\{1-\gamma\left(1+\lambda^{-1}\right)\right\}=2\gamma\left(1-\frac{1}{\xi}\right),
\]
the second equality by the same substitution.
By \eqref{eq:ktbias} the second term of \eqref{eq:ktexp} is at most $n^{2\gamma(1-1/\xi)+o(1)}$, which is the same order.
Hence the sum is $n^{-2\{1-\gamma(1+\lambda^{-1})\}+o(1)}$ and the root mean squared error is $n^{-\{1-\gamma(1+\lambda^{-1})\}+o(1)}$.

If $\nu\le0$ then $v(c_n)=n^{o(1)}$, so $\var\{T_i(c_n)\}\le n^{o(1)}$.
For the lower bound, fix an interval $[a_1,a_2]\subset(\max\{z_0,0\},\infty)$ on which $0<q<1$, which exists because $q\in\RV_{-1/\xi}$ is eventually positive and vanishes at infinity, and note that $a_2<c_n$ for all large $n$.
Given $X_i=u\in(0,c_n)$ the variable $T_i(c_n)$ equals $-1/f_X(u)$ with probability $q(u)$ and zero with probability $p(u)$, so its conditional variance is $q(u)p(u)/f_X(u)^{2}$, and the law of total variance gives
\[
\var\{T_i(c_n)\}\ \ge\ \E\left[\var\{T_i(c_n)\mid X_i\}\right]\ \ge\ \int_{a_1}^{a_2}\frac{q(u)p(u)}{f_X(u)}du\ =:\ C .
\]
The integrand is positive on $[a_1,a_2]$ and the interval does not depend on $n$, so $C>0$ is a constant free of $n$.
Hence $\var\{T_i(c_n)\}/n=n^{-1+o(1)}$.
For the bias, the identity
\[
\gamma\left(1-\frac{1}{\xi}\right)+\frac12=\frac{2\lambda\xi-\lambda+\xi}{2(\lambda+\xi)},\quad
\nu=\frac{2\lambda\xi-\lambda+\xi}{\lambda\xi},
\]
shows that the two share a numerator and so have the same sign, so $\nu\le0$ gives $2\gamma(1-1/\xi)\le-1$ and, by \eqref{eq:ktbias}, $b(c_n)^{2}\le n^{-1+o(1)}$.
The variance term in \eqref{eq:ktexp} is therefore $n^{-1+o(1)}$, while the squared bias is no larger than $n^{-1+o(1)}$. Hence the root mean squared error is $n^{-1/2+o(1)}$.

The two regimes are collected in \eqref{eq:rate}, since $1-\gamma(1+\lambda^{-1})\ge1/2$ if and only if $\gamma\le\lambda/\{2(1+\lambda)\}$, which by the definition of $\gamma$ is $\nu\le0$.
The exponent is positive in the second regime because $1-\gamma(1+\lambda^{-1})>0$ if and only if $\gamma<\lambda/(1+\lambda)$, which by the same definition is $\xi<1$.
\end{proofof}

\begin{proofof}{Theorem \ref{thm:ident}}
Assumption \ref{ass:main}(ii) states $xf_X'(x)/f_X(x)=-(1/\lambda+1+\eta(x))$ with $\eta(x)\to0$, and $f_X'/f_X$ is continuous, so the Karamata representation gives $f_X\in\RV_{-1/\lambda-1}$.

\ppart{a}
Under (iii.a), $q\in\RV_{-1/\xi}$ by hypothesis, so $g=qf_X\in\RV_{-\varrho}$ with $\varrho=1/\xi+1/\lambda+1>1$, and $g$ is positive and continuous, therefore locally bounded.
Karamata's theorem and \eqref{eq:tailrep} give
\[
1-F_{X\mid Y=0}(x)\sim\frac{xg(x)}{(\varrho-1)\Prob(Y_i=0)}\in\RV_{-(1/\xi+1/\lambda)} ,
\]
and since $F_{X\mid Y=0}$ has right endpoint $+\infty$, Gnedenko's theorem \citep[Theorem 1.2.1]{deHaan07} gives $F_{X\mid Y=0}\in\mathcal D(G_\gamma)$ with $1/\gamma=1/\xi+1/\lambda$, that is $\gamma=\lambda\xi/(\lambda+\xi)$.
This is strictly increasing in $\xi$ on $(0,\infty)$ with $\gamma<\lambda$ for every finite $\xi$.
Solving $\gamma(\lambda+\xi)=\lambda\xi$ for $\xi$ gives $\xi=\gamma\lambda/(\lambda-\gamma)$, which is \eqref{eq:gamma}.

\ppart{b}
Under (iii.b), $q$ is $\Gamma$ varying with some auxiliary function $b_q$, and $b_q(t)=o(t)$.
For each real $x$ and $t\to\infty$ we have $t+xb_q(t)=t\{1+o(1)\}$, so $f_X\in\RV_{-1/\lambda-1}$ gives $f_X\{t+xb_q(t)\}/f_X(t)\to1$ by the uniform convergence theorem \citep[Theorem 1.5.2]{Bingham1987}, and therefore
\[
\frac{g\{t+xb_q(t)\}}{g(t)}=\frac{q\{t+xb_q(t)\}}{q(t)}\cdot\frac{f_X\{t+xb_q(t)\}}{f_X(t)}\longrightarrow e^{-x},
\]
that is, $g$ is $\Gamma$ varying with the same auxiliary function $b_q$.
The integration property of $\Gamma$ variation gives that $\int_x^{\infty}g$ is $\Gamma$ varying with the same auxiliary function and $\int_x^{\infty}g\sim b_q(x)g(x)$.
By \eqref{eq:tailrep} the factor $\Prob(Y_i=0)$ cancels in the ratio, so $1-F_{X\mid Y=0}$ is $\Gamma$ varying with auxiliary function $b_q$; the $\Gamma$-variation characterization of the Gumbel domain therefore gives $F_{X\mid Y=0}\in\mathcal D(G_0)$, and $\{1-F_{X\mid Y=0}\}/f_{X\mid Y=0}\sim b_q$.
If in addition $q$ is twice continuously differentiable with $q'<0$ and $a_q'\to0$, then $b_q$ may be taken to be $a_q$, since the von Mises condition makes $q$ be $\Gamma$ varying with auxiliary function $a_q$.

\ppart{c}
Assumption \ref{ass:main}(iii.c) gives $q\in\RV_{0}$ by hypothesis.
Since $f_X\in\RV_{-1/\lambda-1}$, we have $g=qf_X\in\RV_{-1/\lambda-1}$.
Karamata's theorem and \eqref{eq:tailrep} therefore give $1-F_{X\mid Y=0}\in\RV_{-1/\lambda}$, and Gnedenko's theorem gives $F_{X\mid Y=0}\in\mathcal D(G_\lambda)$.

\ppart{d}
Lemma \ref{lem:rapid} applied to $a_X$ gives $a_X(x)/x\to0$, so $-xf_X'(x)/f_X(x)=x/a_X(x)\to\infty$ and Assumption \ref{ass:main}(ii) holds for no $\lambda>0$.
Since $a_X'\to0$, the von Mises condition makes $f_X$ be $\Gamma$ varying with auxiliary function $a_X$.

Under (iii.a) or (iii.c), $q$ is regularly varying, so for each real $x$ the argument $t+xa_X(t)=t\{1+o(1)\}$ gives $q\{t+xa_X(t)\}/q(t)\to1$ by the uniform convergence theorem, and therefore $g=qf_X$ is $\Gamma$ varying with auxiliary function $a_X$, exactly as in part (b).

Under (iii.b), $q$ and $f_X$ are $\Gamma$ varying with auxiliary functions $b_q$ and $a_X$, respectively, both of smaller order than $t$.
Suppose $b_q/a_X\to c\in[0,\infty]$ and set $a=b_qa_X/(b_q+a_X)$, so that $a/b_q\to1/(1+c)$ and $a/a_X\to c/(1+c)$, with the conventions $a\sim b_q$ when $c=0$ and $a\sim a_X$ when $c=\infty$.
Then, by the local uniform convergence theorem for $\Gamma$ variation \citep[Section 1.2]{deHaan07},
\[
\frac{g\{t+xa(t)\}}{g(t)}\longrightarrow e^{-x/(1+c)}e^{-xc/(1+c)}=e^{-x},
\]
so $g$ is $\Gamma$ varying with auxiliary function $a$.
The conclusion also follows by a second route that does not require $b_q/a_X$ to converge.
If $q$ is twice continuously differentiable with $q'<0$ and $a_q'\to0$, then $-g'/g=1/a_q+1/a_X=:1/a_g$ with $a_g=a_qa_X/(a_q+a_X)$ and
\[
a_g'=a_q'\left\{\frac{a_X}{a_q+a_X}\right\}^{2}+a_X'\left\{\frac{a_q}{a_q+a_X}\right\}^{2},
\]
whose two weights lie in $[0,1]$, so $|a_g'|\le|a_q'|+|a_X'|\to0$ and the von Mises condition applies to $g$ directly.
The second route does not imply the first: $a_q(t)=t/\log t$ and $a_X(t)=(t/\log t)\{2+\sin(\log\log t)\}$ satisfy $a_q'\to0$ and $a_X'\to0$, while $a_q/a_X$ oscillates. Conversely, $f_X(x)\propto e^{-x}$ and $q(x)\propto\exp\{-x+\sin(x^2)/x^2\}$ satisfy $a_q/a_X\to1$ but $a_q'\not\to0$, so the first route does not imply the second.

In every case $g$ is $\Gamma$ varying, so the integration property and the characterization of the Gumbel domain, together with \eqref{eq:tailrep}, give that $1-F_{X\mid Y=0}$ is $\Gamma$ varying and $F_{X\mid Y=0}\in\mathcal D(G_0)$.
\end{proofof}

\begin{proofof}{Proposition \ref{prop:multi}}
\ppart{a}
Write $M=m(x)$ and let $C_p=\sup_{x>z_0}\E[|X_i^{a\prime}\beta|^{p}\mid X_i=x]$, which is finite by \eqref{eq:condmom}.
Since $\varepsilon_i\perp(X_i,X_i^{a})$, the event $\{m(X_i)+X_i^{a\prime}\beta\ge\varepsilon_i\}$ gives \eqref{eq:mixture}, that is $q(x)=\int S_\varepsilon(M+u)dG_x(u)$.
Fix $\delta\in(0,1/2)$ and $B>0$ and split the integral over $A_1=\{|u|\le B\}$, $A_2=\{B<|u|\le\delta M\}$, and $A_3=\{|u|>\delta M\}$.
Markov's inequality gives $G_x(|u|>B)\le C_pB^{-p}$ uniformly in $x$, and hence $G_x(A_1)\ge1-C_pB^{-p}$.

On $A_1$, $(M+u)/M\to1$ uniformly over $|u|\le B$, so the uniform convergence theorem for regularly varying functions \citep[Theorem 1.5.2]{Bingham1987} gives $\sup_{|u|\le B}|S_\varepsilon(M+u)/S_\varepsilon(M)-1|\to0$, whence $\int_{A_1}S_\varepsilon(M+u)dG_x(u)=S_\varepsilon(M)\{G_x(A_1)+o(1)\}$.
On $A_2$, we have $S_\varepsilon(M+u)\le S_\varepsilon\{(1-\delta)M\}$ and $S_\varepsilon\{(1-\delta)M\}/S_\varepsilon(M)\to(1-\delta)^{-1/\xi_S}$, so that
$\int_{A_2}S_\varepsilon(M+u)dG_x(u)\le2(1-\delta)^{-1/\xi_S}C_pB^{-p}S_\varepsilon(M)$ for large $x$.
On $A_3$, $S_\varepsilon\le1$ and $G_x(A_3)\le C_p(\delta M)^{-p}$, and since $S_\varepsilon\in\RV_{-1/\xi_S}$ with $p>1/\xi_S$, Potter's bounds \citep[Theorem 1.5.6]{Bingham1987} give $M^{-p}/S_\varepsilon(M)\to0$ as $M\to\infty$, so $\int_{A_3}S_\varepsilon(M+u)dG_x(u)=o\{S_\varepsilon(M)\}$.

Letting $x\to\infty$ and then $B\to\infty$ gives $q(x)/S_\varepsilon(M)\to1$.
Since $S_\varepsilon\in\RV_{-1/\xi_S}$, asymptotic equivalence gives $q\in\RV_{-1/\xi_S}$ when $m(x)=x$, and $q\in\RV_{-\kappa/\xi_S}$ when $m\in\RV_\kappa$ with $\kappa>0$ \citep[Proposition 1.5.7]{Bingham1987}.
These regular-variation conclusions are exactly Assumption \ref{ass:main}(iii.a), with $\xi=\xi_S$ in the first case and $\xi=\xi_S/\kappa$ in the second.

\ppart{b}
Write $M=m(x)$ and $A_m(x)=a_\varepsilon(M)/m'(x)$.
By \eqref{eq:mixture} and \eqref{eq:tailstable}, $q(x)=\int_{\mathcal U}S_\varepsilon\{M+u\}r_x(u)dG(u)$.
Boundedness of $\mathcal U$ and the uniform conditions in \eqref{eq:gumbelshift}-\eqref{eq:tailstable} permit differentiation under the integral.
The definition of $a_\varepsilon$ gives $S_\varepsilon'=-S_\varepsilon/a_\varepsilon$ and $S_\varepsilon''=S_\varepsilon\{1+a_\varepsilon'\}/a_\varepsilon^2$.
Because $a_\varepsilon'(t)\to0$ as $t\to\infty$, Lemma \ref{lem:rapid} gives $a_\varepsilon(t)/t\to0$.
With \eqref{eq:msmooth} this yields $A_m(x)/x=\{a_\varepsilon(M)/M\}\{m(x)/(xm'(x))\}\to0$ and $A_m(x)m''(x)/m'(x)\to0$.
Let $B_x=\int_{\mathcal U}S_\varepsilon\{M+u\}dG(u)$.
Since $\mathcal U$ is bounded and $M\to\infty$, $a_\varepsilon'(t)\to0$ implies $\sup_{u\in\mathcal U}|a_\varepsilon'\{M+u\}|\to0$, and the second condition in \eqref{eq:gumbelshift} gives $a_\varepsilon\{M+u\}=a_\varepsilon(M)\{1+o(1)\}$ uniformly in $u\in\mathcal U$.
Substituting into the derivative formulas and using \eqref{eq:tailstable} and \eqref{eq:msmooth} therefore yields
\[
q(x)=B_x\{1+o(1)\},\quad
q'(x)=-\frac{B_x}{A_m(x)}\{1+o(1)\},\quad
q''(x)=\frac{B_x}{A_m(x)^2}\{1+o(1)\}.
\]
Thus $q'<0$ eventually and $a_q'(x)=q(x)q''(x)/\{q'(x)\}^2-1\to0$, which is the von Mises condition, so $q$ is $\Gamma$ varying with auxiliary function $a_q$ and Assumption \ref{ass:main}(iii.b) holds.
Theorem \ref{thm:ident}(b) then gives $F_{X\mid Y=0}\in\mathcal D(G_0)$ and $\gamma=0$.
\end{proofof}

The following lemma establishes some properties of the likelihood-ratio statistic \eqref{eq:LR}.

\begin{lemma}\label{lem:analytic}
$R_k$ is real analytic and nonconstant on the interior of $\mathcal V$, and the distribution function of $R_k(\mathbf V^{*})$ under $\gamma=0$ is continuous on $\Real$.
\end{lemma}

\begin{proof}
Fix $\gamma\in(0,\bar\gamma]$ and a compact subset $K$ of the interior of $\mathcal V$, on which $\min_{j\le k-1}v_j^{*}\ge c_0$ for some $c_0>0$.
For each $t>0$ the integrand in \eqref{eq:vstar} extends to each coordinate of $v^{*}$ in the complex polydisc of radius $c_0/2$ about $K$, on which $\mathrm{Re}(1+\gamma v_j^{*}t)\ge1+\gamma c_0t/2>0$, so the principal branch of $(1+\gamma v_j^{*}t)^{-(1+1/\gamma)}$ is holomorphic there and has modulus at most $(1+\gamma c_0t/2)^{-(1+1/\gamma)}$.
The integrand is therefore bounded in modulus by $t^{k-2}(1+\gamma c_0t/2)^{-(k-1)(1+1/\gamma)}$, which is integrable in $t$ because $(k-1)(1+1/\gamma)-(k-2)>1$.
Morera's theorem applied along any closed contour in that polydisc, with the contour integral and the integral in $t$ interchanged by Fubini's theorem, gives that $f_{\mathbf V^{*}\mid\gamma}$ is holomorphic there and hence real analytic on the interior of $\mathcal V$.

The same argument applied to the integral over $\gamma$ against the bounded weight $w$ gives that $\bar f_w=\int_0^{\bar\gamma}f_{\mathbf V^{*}\mid\gamma}w(\gamma)d\gamma$ is real analytic.
The denominator \eqref{eq:vstarnull} is real analytic and strictly positive on the interior of $\mathcal V$, so $R_k=\bar f_w/f_{\mathbf V^{*}\mid0}$ is real analytic.

It remains to show that $R_k$ is nonconstant.
For $\epsilon\in(0,1)$ let $v^{*}(\epsilon)$ have $v_1^{*}=1$, $v_j^{*}=\epsilon$ for $2\le j\le k-1$ and $v_k^{*}=0$, a point of the interior of $\mathcal V$.
The integrand in \eqref{eq:vstar} increases as $\epsilon$ decreases, so monotone convergence, first in $t$ and then in $\gamma$, gives
\begin{equation}\label{eq:corner}
\lim_{\epsilon\downarrow0}\bar f_w\{v^{*}(\epsilon)\}=\int_0^{\bar\gamma}\Lambda_k(\gamma)w(\gamma)d\gamma,\quad
\Lambda_k(\gamma)=(k-1)!\int_0^{\infty}t^{k-2}(1+\gamma t)^{-(1+1/\gamma)}dt .
\end{equation}
The substitution $s=\gamma t$ gives $\Lambda_k(\gamma)=(k-1)!\gamma^{-(k-1)}B\{k-1,1/\gamma-(k-2)\}$ for $\gamma<1/(k-2)$ and $\Lambda_k(\gamma)=+\infty$ for $\gamma\ge1/(k-2)$.
Moreover $B(k-1,z)=\Gamma(k-1)\Gamma(z)/\Gamma(k-1+z)\sim1/z$ as $z\downarrow0$, so $\Lambda_k(\gamma)\asymp\{1/(k-2)-\gamma\}^{-1}$ as $\gamma\uparrow1/(k-2)$.
Since $1/(k-2)\le1=\bar\gamma$ for every $k\ge3$ and $w$ is uniform on $(0,1]$, the integral in \eqref{eq:corner} is $+\infty$.
By \eqref{eq:vstarnull}, $f_{\mathbf V^{*}\mid0}\{v^{*}(\epsilon)\}\to(k-1)!(k-2)!\in(0,\infty)$, so $R_k\{v^{*}(\epsilon)\}\to\infty$ and $R_k$ is not constant.
A nonconstant real analytic function on a connected open set has level sets of Lebesgue measure zero, so $\Prob_{0}\{R_k(\mathbf V^{*})=c\}=0$ for every $c\in\Real$, which is the stated continuity.
\end{proof}

\begin{proofof}{Theorem \ref{thm:size}}
The hypothesis $F_{X\mid Y=0}\in\mathcal D(G_0)$ is assumed directly, and it holds under Assumption \ref{ass:main} with (iii.b) by Theorem \ref{thm:ident}(b) and under the hypotheses of Theorem \ref{thm:ident}(d) by that part.
Because $\{(Y_i,X_i)\}_{i=1}^n$ are i.i.d., conditional on $(Y_1,\ldots,Y_n)$ the variables $\{X_i:Y_i=0\}$ are independent with common distribution $F_{X\mid Y=0}$, and their number is $n_0$.
Let $\mathcal A=\{v^{*}\in\mathcal V:R_k(v^{*})>\mathrm{cv}(\alpha,k)\}$, set $h_\ell=\Prob(\mathbf X^{*(0)}\in\mathcal A\mid n_0=\ell)$ for $\ell\ge k$ and $h_\ell=0$ for $\ell<k$, and note that the conditional probability of rejection given $(Y_1,\ldots,Y_n)$ equals $h_{n_0}$, a function of $n_0$ alone.

The extreme value theorem for the top $k$ order statistics of $\ell$ i.i.d.\ draws from a law in $\mathcal D(G_0)$, together with the continuous mapping theorem applied to \eqref{eq:Xstar}, gives $\mathbf X^{*(0)}\dto\mathbf V^{*}$ under $\gamma=0$ as $\ell\to\infty$.
The boundary of $\mathcal A$ within $\mathcal V$ is contained in $\{v^{*}:R_k(v^{*})=\mathrm{cv}(\alpha,k)\}\cup\partial\mathcal V$, and $\mathbf V^{*}$ has a density on the interior of $\mathcal V$ by \eqref{eq:vstar}, so $\Prob_0(\mathbf V^{*}\in\partial\mathcal V)=0$, while $\Prob_0\{R_k(\mathbf V^{*})=\mathrm{cv}(\alpha,k)\}=0$ by Lemma \ref{lem:analytic}.
The portmanteau theorem therefore gives $h_\ell\to\Prob_0(\mathbf V^{*}\in\mathcal A)=\alpha$ as $\ell\to\infty$, the last equality being the definition of $\mathrm{cv}(\alpha,k)$ together with the same lemma.

Since $\Prob(Y_i=0)\in(0,1)$, the strong law gives $n_0/n\to\Prob(Y_i=0)$ almost surely, hence $n_0\to\infty$ almost surely and $h_{n_0}\to\alpha$ almost surely.
As $0\le h_{n_0}\le1$, bounded convergence gives $\Prob(\varphi_k=1)=\E[h_{n_0}]\to\alpha$.
\end{proofof}

\begin{proofof}{Proposition \ref{prop:panel}}
\ppart{a}
Fix $t$.
The pairs $\{(Y_{it},X_{it})\}_{i=1}^{n}$ are i.i.d.\ across $i$ because the individual histories are, and $q_t$ is a functional of their joint distribution.
Assumption \ref{ass:main}(i) holds for this pair by hypothesis, so Theorem \ref{thm:ident} applies to the population objects $(F_{X_t},q_t)$ under the corresponding part of Assumption \ref{ass:main}(iii).
Under part (iii.b), Theorem \ref{thm:ident}(b) gives $F_{X_t\mid Y_t=0}\in\mathcal D(G_0)$, so Theorem \ref{thm:size} applies to the period-$t$ subsample.
The conditional distribution of $\alpha_i$ or of $Y_{i,t-1}$ given $X_{it}=x$ enters the argument only through $q_t(x)$, over which it is integrated, so no property of that conditional distribution is used.

\ppart{b}
Condition on $(Y_1,\ldots,Y_n)$, the outcomes of the period in question.
The variables $\{X_i:Y_i=0\}$ and $\{X_i:Y_i=1\}$ are then independent collections indexed by disjoint sets, so the right tail statistic and the left tail statistic are conditionally independent.
Each converges in distribution to a limit that does not depend on the conditioning variables, by the argument in the proof of Theorem \ref{thm:size} applied to the two subsamples, so the joint limit is the product measure and $p_L$ and $p_R$ are asymptotically independent and uniform on $(0,1)$ under the intersection null.
Hence $\Prob[\min\{p_L,p_R\}\ge1-(1-\alpha)^{1/2}]\to\{(1-\alpha)^{1/2}\}^{2}=1-\alpha$.

\ppart{c}
By part (b), each $p_t^{S}$ is asymptotically uniform on $(0,1)$ under the intersection null. Therefore,
\[
\limsup_{n\to\infty}\Prob\left\{\min_{t\le T}p_t^{S}<\frac{\alpha}{T}\right\}
\le\sum_{t=1}^{T}\limsup_{n\to\infty}\Prob\left\{p_t^{S}<\frac{\alpha}{T}\right\}
=T\frac{\alpha}{T}=\alpha
\]
by Boole's inequality. No restriction on dependence across periods is required.
\end{proofof}

\begin{proofof}{Proposition \ref{prop:power}}
\ppart{a}
The group of maps $x\mapsto ax+b$ with $a>0$ acts on $\mathbf V$ by $\mathbf V\mapsto a\mathbf V+b\mathbf 1$, the statistic $\mathbf V^{*}$ in \eqref{eq:Xstar} is invariant under this action, and two points of $\{v:v_k<v_1\}$ with the same value of $\mathbf V^{*}$ lie in the same orbit, so $\mathbf V^{*}$ is a maximal invariant and every invariant test is a measurable function of it \citep[Chapter 6]{lehmann2005testing}.
Restricted to such tests the constraint is $\E_0[\varphi]\le\alpha$ and the objective is
$\int_0^{\bar\gamma}\E_\gamma[\varphi]w(\gamma)d\gamma=\E_{\bar f_w}[\varphi]$, where $\bar f_w=\int_0^{\bar\gamma}f_{\mathbf V^{*}\mid\gamma}w(\gamma)d\gamma$ is a probability density on $\mathcal V$ by Tonelli's theorem.
The Neyman-Pearson lemma applied to the simple null $f_{\mathbf V^{*}\mid0}$ against the simple alternative $\bar f_w$ shows that the test rejecting when $\bar f_w/f_{\mathbf V^{*}\mid0}$, which is $R_k$, exceeds its $1-\alpha$ null quantile is most powerful at level $\alpha$, the quantile being attained exactly by Lemma \ref{lem:analytic}.
The same lemma applied against $f_{\mathbf V^{*}\mid\gamma_1}$ for each fixed $\gamma_1>0$ gives the envelope.

\ppart{b}
For every $\gamma\ge0$ the integrand in \eqref{eq:vstar} is strictly positive for every $t>0$ and every $v^{*}$ in the interior of $\mathcal V$, so $f_{\mathbf V^{*}\mid\gamma}>0$ there and $\Prob_\gamma$ and $\Prob_0$ are mutually absolutely continuous on $\mathcal V$.
The acceptance region $\mathcal V\setminus\mathcal A$ has $\Prob_0$ measure $1-\alpha>0$, hence positive Lebesgue measure, hence positive $\Prob_\gamma$ measure, so $\pi_k(\gamma)=1-\Prob_\gamma(\mathcal V\setminus\mathcal A)<1$.
\end{proofof}

\begin{proofof}{Theorem \ref{thm:uniform}}
Throughout, $\Gamma_j=\sum_{i\le j}E_i$ with $E_i$ i.i.d.\ standard exponential.

\ppart{a}
Fix $\gamma_{0}\in(0,\lambda)$ and $M>1$, and let $1-F_X(x)=x^{-1/\lambda}$ on $[1,\infty)$, so that $f_X(x)=\lambda^{-1}x^{-1/\lambda-1}$.
Define $H_M$ by joining a Pareto law of index $1/\gamma_{0}$ to an exponential tail of scale $\gamma_{0}M$,
\[
1-H_M(x)=\begin{cases}
x^{-1/\gamma_{0}}, & 1\le x\le M,\\
M^{-1/\gamma_{0}}\exp\{-(x-M)/(\gamma_{0}M)\}, & x>M.
\end{cases}
\]
Differentiating gives a density $h_M$ that is continuous at $M$.
Set $q_M=p h_M/f_X$, so that $F_{X\mid Y=0}=H_M$ and $\Prob(Y_i=0)=p$.
On $[1,M]$, $h_M/f_X=(\lambda/\gamma_{0})x^{1/\lambda-1/\gamma_{0}}\le\lambda/\gamma_{0}$.
For $x=My>M$, the same ratio is at most $(\lambda/\gamma_{0})C_{\lambda,\gamma_{0}}$, where $C_{\lambda,\gamma_{0}}=\sup_{y\ge1}y^{1+1/\lambda}e^{-(y-1)/\gamma_{0}}<\infty$.
Taking $p=\gamma_{0}/(2\lambda C_{\lambda,\gamma_{0}})$ therefore yields $q_M\le1$ for every $M\ge1$.
For $x>M$, $q_M$ is $\Gamma$ varying with auxiliary function $\gamma_{0}M$, so Assumption \ref{ass:main}(iii.b) holds and $F_M\in\mathcal F_{0}$ for every $M$.

Let $P$ denote the Pareto law with $1-P(x)=x^{-1/\gamma_{0}}$ on $[1,\infty)$, and note that $H_M$ and $P$ agree on $[1,M]$.
Let $\zeta_1,\ldots,\zeta_\ell$ be i.i.d.\ $P$, whose tail quantile function is $U_P(t)=(1/\{1-P\})^{\leftarrow}(t)=t^{\gamma_{0}}$.
By the quantile transform, $\zeta_{\ell:\ell-j+1}\stackrel{d}{=}U_P(\Gamma_{\ell+1}/\Gamma_j)=\Gamma_{\ell+1}^{\gamma_{0}}\Gamma_j^{-\gamma_{0}}$ jointly in $j\le k$, so the common factor $\Gamma_{\ell+1}^{\gamma_{0}}$ cancels in \eqref{eq:Xstar} and the self-normalized vector equals $\mathbf V^{*}$ at $\gamma=\gamma_{0}$ exactly, for every $\ell$.
Hence $\Prob_{P}(\varphi_k=1\mid n_0=\ell)=\pi_k(\gamma_{0})$ for every $\ell\ge k$.
Since $H_M$ and $P$ agree on $[1,M]$, their total variation distance is at most $1-P(M)=M^{-1/\gamma_{0}}$, and by subadditivity over $\ell$ independent draws the two conditional samples are within $\ell M^{-1/\gamma_{0}}$ in total variation, so
\begin{equation}\label{eq:tvbound}
\left|\Prob_{F_M}(\varphi_k=1\mid n_0=\ell)-\pi_k(\gamma_{0})\right|\le\ell M^{-1/\gamma_{0}} .
\end{equation}
The order of the two limits matters, and the supremum is taken at each $n$ before the liminf.
Fix $n$.
Under every member of the family $\Prob(Y_i=0)=p$, so the law of $n_0$ is binomial with parameters $n$ and $p$ and does not depend on $M$, and $n_0\le n$.
The right side of \eqref{eq:tvbound} is therefore at most $nM^{-1/\gamma_{0}}$ uniformly in $\ell$, and letting $M\to\infty$ at fixed $n$ gives $\Prob_{F_M}(\varphi_k=1)\to\pi_k(\gamma_{0})\Prob(n_0\ge k)$.
Since $F_M\in\mathcal F_{0}$ for every $M$, it follows that $\sup_{F\in\mathcal F_{0}}\Prob_F(\varphi_k=1)\ge\pi_k(\gamma_{0})\Prob(n_0\ge k)$ at every $n$, and $\Prob(n_0\ge k)\to1$ then gives $\liminf_{n}\sup_{F\in\mathcal F_{0}}\Prob_F(\varphi_k=1)\ge\pi_k(\gamma_{0})$.
Equivalently, the offending sequence may be exhibited directly by taking $M_n=n^{2\gamma_{0}}$, for which $nM_n^{-1/\gamma_{0}}=n^{-1}$ and $\Prob_{F_{M_n}}(\varphi_k=1)=\pi_k(\gamma_{0})+o(1)$.
The reverse order is not available: at fixed $M$ each $F_M$ lies in the null, so $\Prob_{F_M}(\varphi_k=1)\to\alpha$ by Theorem \ref{thm:size}, which is why $M$ must grow with $n$.
Since $\gamma_{0}>0$ was arbitrary and $\lambda>\gamma_{0}$ may be chosen freely in the construction, the first inequality follows.
The second holds because $\int_0^{\bar\gamma}\pi_k(\gamma)w(\gamma)d\gamma>\alpha$ by Proposition \ref{prop:power}(a) and the Neyman-Pearson lemma applied to the nonconstant ratio $\bar f_w/f_{\mathbf V^{*}\mid0}$, so $\pi_k(\gamma)>\alpha$ on a set of positive $w$-measure.

\ppart{b}
By the quantile transform, conditional on $n_0=\ell\ge k$,
$X^{(0)}_{\ell:\ell-j+1}\stackrel{d}{=}U_F(\Gamma_{\ell+1}/\Gamma_j)$ jointly in $j\le k$.
Put $t=\Gamma_{\ell+1}$ and, on the event $\{\Gamma_{\ell+1}\ge t_{0}\}$, apply the bound in \eqref{eq:uniclass} at $x=1/\Gamma_j$. This choice is admissible because $\Gamma_j<\Gamma_{\ell+1}=t$, and hence $1/\Gamma_j>1/t$, so
\begin{equation}\label{eq:unibound}
\max_{j\le k}\left|\widetilde V_j-V_j\right|\ \le\ A_{0}\Gamma_{\ell+1}^{\rho}\max_{j\le k}\left(1+|\log\Gamma_j|\right)\ =:\ \Delta_{\ell},
\end{equation}
where $\widetilde V_j=\{X^{(0)}_{\ell:\ell-j+1}-U_F(t)\}/a_F(t)$ and $V_j=-\log\Gamma_j$.
The right side of \eqref{eq:unibound} depends on $F$ only through the constants $(\rho,A_{0})$, and $\Gamma_{\ell+1}^{\rho}\max_{j\le k}(1+|\log\Gamma_j|)$ has a law free of $F$ with $\Gamma_{\ell+1}^{\rho}\to0$ almost surely, while $\Prob(\Gamma_{\ell+1}<t_{0})\to0$ as $\ell\to\infty$ and is also free of $F$, so
$\sup_{F}\Prob(\Delta_{\ell}>\delta\mid n_0=\ell)\le\eta_{\ell}(\delta)$ with $\eta_{\ell}(\delta)\to0$ for each $\delta>0$.
Since $\Prob(Y_i=0)\ge p_{0}$, $n_0\ge p_{0}n/2$ with probability tending to one uniformly in $F$, so $\sup_F\Prob(\Delta_{n_0}>\delta)\to0$.

The map \eqref{eq:Xstar} is invariant to the location shift $U_F(t)$ and the scale $a_F(t)>0$, so $\mathbf X^{*(0)}=(\widetilde{\mathbf V}-\widetilde V_k)/(\widetilde V_1-\widetilde V_k)$ with $\widetilde{\mathbf V}=(\widetilde V_1,\ldots,\widetilde V_k)$, while $\mathbf V^{*}=(\mathbf V-V_k)/(V_1-V_k)$ with the same $\Gamma_j$.
Fix $\varsigma>0$ and work on $\mathcal E_{\varsigma,\delta}=\{V_1-V_k\ge\varsigma\}\cap\{\Delta_{n_0}\le\delta\}$ with $\delta\le\varsigma/4$.
On $\mathcal E_{\varsigma,\delta}$, $\widetilde V_1-\widetilde V_k\ge V_1-V_k-2\Delta_{n_0}\ge\varsigma/2$, and for each $j\le k$,
\begin{align*}
\left|\frac{\widetilde V_j-\widetilde V_k}{\widetilde V_1-\widetilde V_k}-\frac{V_j-V_k}{V_1-V_k}\right|
\le &\frac{\left|(\widetilde V_j-\widetilde V_k)-(V_j-V_k)\right|}{\widetilde V_1-\widetilde V_k}\\
+ &\frac{V_j-V_k}{V_1-V_k}\cdot\frac{\left|(V_1-V_k)-(\widetilde V_1-\widetilde V_k)\right|}{\widetilde V_1-\widetilde V_k}
\ \le\ \frac{8\Delta_{n_0}}{\varsigma},
\end{align*}
using $|\widetilde V_j-\widetilde V_k-(V_j-V_k)|\le2\Delta_{n_0}$, $0\le(V_j-V_k)/(V_1-V_k)\le1$ and $\widetilde V_1-\widetilde V_k\ge\varsigma/2$.
Hence $\|\mathbf X^{*(0)}-\mathbf V^{*}\|_\infty\le8\Delta_{n_0}/\varsigma\le8\delta/\varsigma$ on $\mathcal E_{\varsigma,\delta}$.
The event is defined through $\mathbf V^{*}$ and $\Delta_{n_0}$ only, so no property of $\mathbf X^{*(0)}$ is presumed.

The function $R_k$ is not bounded on $\mathcal V$, since the proof of Lemma \ref{lem:analytic} exhibits points of the interior at which it diverges, so the comparison is made on a compact subset.
For $\iota>0$ put $\mathcal V_\iota=\{v^{*}\in\mathcal V:\min_{j\le k-1}(v_j^{*}-v_{j+1}^{*})\ge\iota\}$, which is a compact subset of the interior of $\mathcal V$ on which $R_k$ is real analytic by Lemma \ref{lem:analytic} and therefore Lipschitz with a finite constant $L_\iota$.
Since $V_1>\cdots>V_k$ almost surely, $\Prob_{0}(\mathbf V^{*}\notin\mathcal V_{2\iota})\to0$ as $\iota\downarrow0$.
Restrict to $\delta\le\varsigma\iota/16$, so that $8\delta/\varsigma\le\iota/2$.
On $\mathcal E_{\varsigma,\delta}\cap\{\mathbf V^{*}\in\mathcal V_{2\iota}\}$ the consecutive gaps of $\mathbf X^{*(0)}$ exceed $2\iota-2(\iota/2)=\iota$, so both vectors lie in $\mathcal V_\iota$ and
$|R_k(\mathbf X^{*(0)})-R_k(\mathbf V^{*})|\le 8L_\iota\delta/\varsigma$.
Comparing the rejection events therefore gives
\begin{align*}
\left|\Prob_F(\varphi_k=1)-\alpha\right|
\le \Prob_{0}\left\{\left|R_k(\mathbf V^{*})-\mathrm{cv}(\alpha,k)\right|\le 8L_\iota\delta/\varsigma\right\}\\
+\Prob\left(\Delta_{n_0}>\delta\right)+\Prob(V_1-V_k<\varsigma)+\Prob_{0}(\mathbf V^{*}\notin\mathcal V_{2\iota}).
\end{align*}
The bound is uniform in $F\in\mathcal F_{0}(\rho,A_{0},t_{0},p_{0})$.
Taking $\limsup_n$ at fixed $(\delta,\varsigma,\iota)$ removes the second term.
Letting $\delta\to0$ at fixed $(\varsigma,\iota)$ removes the first, because $\Prob_{0}\{|R_k(\mathbf V^{*})-\mathrm{cv}(\alpha,k)|\le\epsilon\}\downarrow\Prob_{0}\{R_k(\mathbf V^{*})=\mathrm{cv}(\alpha,k)\}=0$ as $\epsilon\downarrow0$ by Lemma \ref{lem:analytic}.
Letting $\varsigma\to0$ removes the third, since $V_1>V_k$ almost surely, and letting $\iota\to0$ removes the fourth.
\end{proofof}

\section{Comparison with the Pickands Test}\label{app:pickands}

Section \ref{sec:power} notes that among standard tail index estimators only that of \citet{Pickands1975} is invariant to location and scale, and therefore a function of the maximal invariant $\mathbf V^{*}$, so that a test built on it competes with $\varphi_k$ on the same terms.
Figure \ref{fig:power} reports the asymptotic power of the two tests in the limit experiment.
The Pickands critical value is set by simulation under $\gamma=0$ rather than taken from its asymptotic normal approximation, so that the comparison is one of power and not of size.
The deficiency is large.
At $k=50$ the Pickands test attains $0.254$ at $\gamma=0.45$ and $0.630$ at $\gamma=1$, against $0.810$ and $0.995$ for $\varphi_k$, and the gap widens as $k$ falls.
The naive version of the Pickands test, which uses the asymptotic standard error $\sqrt3/\{2(\log2)^{2}\sqrt{j_k}\}$ with $j_k=\lfloor k/4\rfloor$, has size $0.044$ at $k=25$ and $0.040$ at $k=50$, so in the limit experiment its deficiency is one of power rather than of size.
The reason is that the Pickands statistic uses three order statistics and discards the rest, whereas $R_k$ uses the whole vector $\mathbf V^{*}$.

\begin{figure}[tp]
\centering
\includegraphics[width=\textwidth]{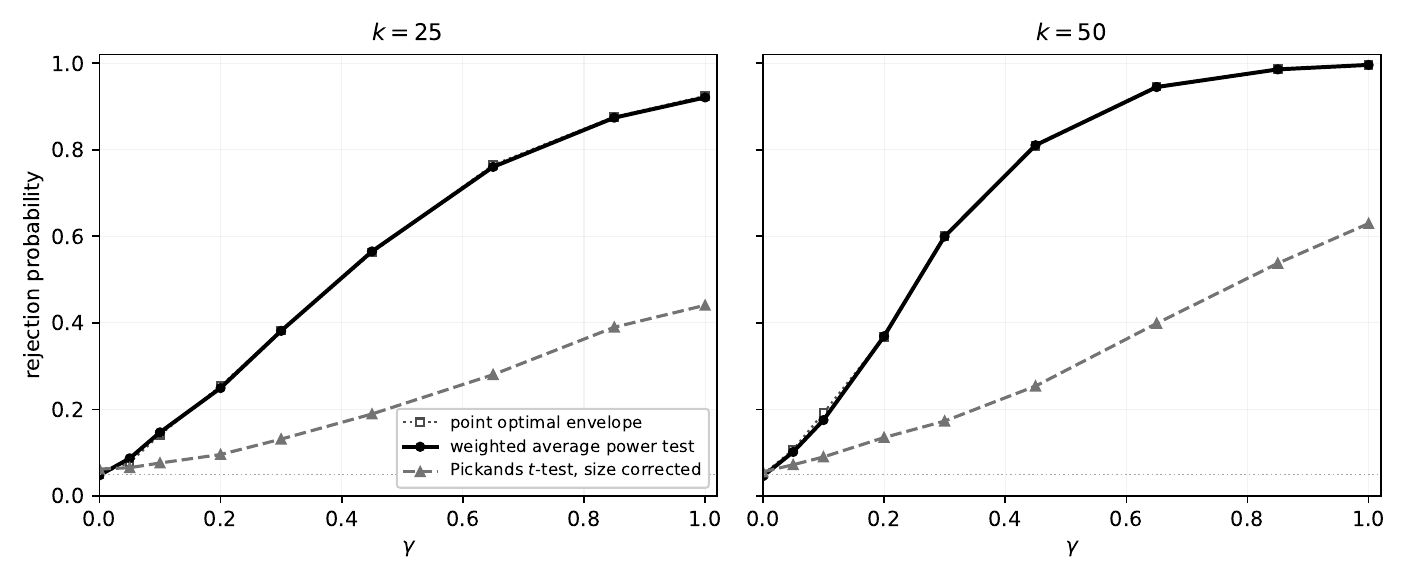}
\caption{Asymptotic power at $\alpha=0.05$. The Pickands statistic is $\{\log2\}^{-1}\log[(V_{j_k}^{*}-V_{2j_k}^{*})/(V_{2j_k}^{*}-V_{4j_k}^{*})]$ with $j_k=\lfloor k/4\rfloor$, and its critical value is simulated under $\gamma=0$ so that the comparison is one of power alone. Based on 2,500 draws.}
\label{fig:power}
\end{figure}

\section{Additional Tables}\label{app:sim}

This appendix collects two sets of tables.
Table \ref{tab:cv} reports the critical values of the test \eqref{eq:LR} for the values of $k$ used in the paper and for ten significance levels, obtained by simulating the null law of $R_k$ under $\gamma=0$ with 200,000 draws for each $k$.
A practitioner needs only this table and the largest $k$ order statistics of the outcome-specific subsample to carry out the test.
Tables \ref{tab:panel_nondynamic} and \ref{tab:panel_dynamic} report the panel simulation results described in Section \ref{sec:simu:panel}, for the static and the dynamic design respectively, with the two tails combined within each period by the \v{S}id\'ak rule and the two periods combined by the Bonferroni rule of Proposition \ref{prop:panel}(c).

\begin{table}[!ht]
\centering
\begin{tabular}{c c c c c c c c c c c}
\hline\hline
$k \backslash \alpha$ & 0.10 & 0.09 & 0.08 & 0.07 & 0.06 & 0.05 & 0.04 & 0.03 & 0.02 & 0.01 \\
\hline
10 & 1.50 & 1.58 & 1.69 & 1.82 & 2.00 & 2.21 & 2.51 & 2.94 & 3.73 & 5.76 \\
25 & 1.18 & 1.27 & 1.39 & 1.52 & 1.70 & 1.95 & 2.29 & 2.83 & 3.81 & 6.59 \\
50 & 0.80 & 0.88 & 0.96 & 1.07 & 1.21 & 1.40 & 1.69 & 2.13 & 2.95 & 5.29 \\
70 & 0.67 & 0.72 & 0.80 & 0.89 & 1.00 & 1.16 & 1.39 & 1.76 & 2.49 & 4.42 \\
100 & 0.54 & 0.59 & 0.65 & 0.72 & 0.82 & 0.96 & 1.15 & 1.48 & 2.10 & 3.73 \\
\hline
\end{tabular}
\caption{Critical values of the test \eqref{eq:LR} with different $k$ and significance level $\alpha$.
Based on 200,000 simulation draws for each $k$.}\label{tab:cv}
\end{table}

\begin{table}[tp]
\centering
\resizebox{\textwidth}{!}{
\begin{tabular}{lcrrrrrrrrrrrr}
\hline\hline
Dist.\ of $\varepsilon$ & $n$ & \multicolumn{3}{c}{$k=10$} & \multicolumn{3}{c}{$k=25$} & \multicolumn{3}{c}{$k=50$} & \multicolumn{3}{c}{$k=70$} \\
 &  & Left & Right & Both & Left & Right & Both & Left & Right & Both & Left & Right & Both \\
\hline
Normal   & 1000 & 0.04 & 0.03 & 0.04 & 0.02 & 0.02 & 0.01 & 0.01 & 0.01 & 0.01 & 0.00 & 0.00 & 0.00 \\
         & 2000 & 0.03 & 0.03 & 0.03 & 0.01 & 0.02 & 0.02 & 0.01 & 0.01 & 0.01 & 0.01 & 0.01 & 0.01 \\
         & 5000 & 0.03 & 0.03 & 0.03 & 0.02 & 0.03 & 0.02 & 0.01 & 0.01 & 0.01 & 0.01 & 0.02 & 0.01 \\
Logistic & 1000 & 0.04 & 0.05 & 0.05 & 0.06 & 0.05 & 0.06 & 0.05 & 0.05 & 0.05 & 0.03 & 0.03 & 0.03 \\
         & 2000 & 0.03 & 0.05 & 0.05 & 0.05 & 0.05 & 0.05 & 0.06 & 0.05 & 0.07 & 0.06 & 0.05 & 0.06 \\
         & 5000 & 0.04 & 0.04 & 0.04 & 0.06 & 0.06 & 0.05 & 0.08 & 0.06 & 0.07 & 0.06 & 0.07 & 0.08 \\ \hline
$t(2)$   & 1000 & 0.13 & 0.13 & 0.16 & 0.16 & 0.17 & 0.22 & 0.19 & 0.20 & 0.27 & 0.16 & 0.16 & 0.22 \\
         & 2000 & 0.14 & 0.15 & 0.18 & 0.23 & 0.23 & 0.29 & 0.30 & 0.29 & 0.39 & 0.30 & 0.30 & 0.41 \\
         & 5000 & 0.15 & 0.18 & 0.21 & 0.30 & 0.31 & 0.40 & 0.40 & 0.41 & 0.54 & 0.48 & 0.45 & 0.61 \\
$t(1)$   & 1000 & 0.22 & 0.22 & 0.29 & 0.40 & 0.40 & 0.55 & 0.50 & 0.49 & 0.68 & 0.52 & 0.52 & 0.69 \\
         & 2000 & 0.24 & 0.24 & 0.32 & 0.46 & 0.46 & 0.60 & 0.64 & 0.64 & 0.81 & 0.71 & 0.70 & 0.86 \\
         & 5000 & 0.25 & 0.25 & 0.34 & 0.47 & 0.47 & 0.64 & 0.73 & 0.73 & 0.88 & 0.83 & 0.83 & 0.95 \\
\hline
\end{tabular}}
\caption{Simulation results for static panel data with $T=2$. Entries are the rejection rates of our proposed test \eqref{eq:LR}. The normal and logistic distributions correspond to the null hypothesis, and the Student-$t$ distributions correspond to the alternative. Significance level is $5\%$. Based on 2,000 simulation draws. The Left and Right columns combine the two period-specific $p$-values within the indicated tail using Bonferroni. The Both columns first combine the left- and right-tail $p$-values within each period using \v{S}id\'ak and then combine the two period-level $p$-values using Bonferroni.}
\label{tab:panel_nondynamic}
\end{table}

\begin{table}[tp]
\centering
\resizebox{\textwidth}{!}{%
\begin{tabular}{lcrrrrrrrrrrrr}
\hline\hline
Dist.\ of $\varepsilon$ & $n$ & \multicolumn{3}{c}{$k=10$} & \multicolumn{3}{c}{$k=25$} & \multicolumn{3}{c}{$k=50$} & \multicolumn{3}{c}{$k=70$} \\
 &  & Left & Right & Both & Left & Right & Both & Left & Right & Both & Left & Right & Both \\
\hline
Normal   & 1000 & 0.03 & 0.03 & 0.03 & 0.02 & 0.02 & 0.01 & 0.00 & 0.01 & 0.01 & 0.00 & 0.01 & 0.01 \\
         & 2000 & 0.03 & 0.03 & 0.03 & 0.02 & 0.02 & 0.02 & 0.01 & 0.01 & 0.01 & 0.00 & 0.01 & 0.01 \\
         & 5000 & 0.03 & 0.03 & 0.02 & 0.03 & 0.02 & 0.02 & 0.02 & 0.01 & 0.01 & 0.01 & 0.01 & 0.01 \\
Logistic & 1000 & 0.05 & 0.04 & 0.05 & 0.05 & 0.04 & 0.04 & 0.02 & 0.05 & 0.04 & 0.01 & 0.05 & 0.03 \\
         & 2000 & 0.05 & 0.04 & 0.05 & 0.06 & 0.05 & 0.06 & 0.05 & 0.04 & 0.05 & 0.04 & 0.06 & 0.05 \\
         & 5000 & 0.05 & 0.04 & 0.05 & 0.06 & 0.04 & 0.06 & 0.07 & 0.05 & 0.07 & 0.08 & 0.05 & 0.07 \\ \hline
$t(2)$   & 1000 & 0.11 & 0.10 & 0.13 & 0.17 & 0.15 & 0.21 & 0.13 & 0.14 & 0.20 & 0.10 & 0.14 & 0.16 \\
         & 2000 & 0.16 & 0.13 & 0.17 & 0.25 & 0.19 & 0.29 & 0.27 & 0.22 & 0.33 & 0.26 & 0.22 & 0.34 \\
         & 5000 & 0.18 & 0.16 & 0.23 & 0.29 & 0.26 & 0.36 & 0.41 & 0.32 & 0.50 & 0.45 & 0.35 & 0.56 \\
$t(1)$   & 1000 & 0.25 & 0.22 & 0.31 & 0.38 & 0.35 & 0.47 & 0.44 & 0.43 & 0.60 & 0.45 & 0.45 & 0.63 \\
         & 2000 & 0.24 & 0.26 & 0.34 & 0.43 & 0.45 & 0.58 & 0.61 & 0.60 & 0.78 & 0.71 & 0.62 & 0.84 \\
         & 5000 & 0.26 & 0.25 & 0.33 & 0.47 & 0.51 & 0.64 & 0.73 & 0.73 & 0.88 & 0.83 & 0.82 & 0.95 \\
\hline
\end{tabular}}
\caption{Simulation results for dynamic panel data with $T=2$. Entries are the rejection rates of our proposed test \eqref{eq:LR}. The normal and logistic distributions correspond to the null hypothesis, and the Student-$t$ distributions correspond to the alternative. Significance level is $5\%$. Based on 2,000 simulation draws. The Left and Right columns combine the two period-specific $p$-values within the indicated tail using Bonferroni. The Both columns first combine the left- and right-tail $p$-values within each period using \v{S}id\'ak and then combine the two period-level $p$-values using Bonferroni.}
\label{tab:panel_dynamic}
\end{table}

\section{The Plug-in Index}\label{app:plugin}

The test in the paper is applied to a covariate that is known, or argued, to dominate the tail of the index.
This appendix removes that requirement.
It shows that the test may instead be applied to a fitted index formed from a preliminary distribution-free estimate of the coefficients, and that the estimation error is asymptotically negligible for the extreme order statistics under a condition relating the tail of the regressor vector to the tail of the index.

Proposition \ref{prop:multi} shows that a single covariate suffices when it dominates the tail of the index.
When no coordinate dominates, the procedure may be applied to a fitted index. Specifically, consider the model $Y_i=I(\widetilde X_i'\delta-\varepsilon_i>0)$, where $\widetilde X_i=(X_i,X_i^a)'$ and $\delta=(1,\beta')'$.
Write the index as $W_i=\widetilde X_i'\delta$, estimate $\delta$ by a distribution-free method such as the maximum score estimator of \citet{Manski1975}, and form $\widehat W_i=\widetilde X_i'\widehat\delta$.

\begin{proposition}\label{prop:plugin}
Suppose $(Y_i,\widetilde X_i)$ is an i.i.d.\ sample with $\Prob(Y_i=0)>0$, let $W_i=\widetilde X_i'\delta$, and suppose the pair $(F_W,q_W)$ with $q_W(w)=\Prob(Y_i=0\mid W_i=w)$ satisfies Assumption \ref{ass:main}, so that $F_{W\mid Y=0}$ is in the domain of attraction of an extreme value distribution with tail index $\gamma$ and normalizing sequences $a_{n_0}>0$ and $b_{n_0}$.
Suppose
\begin{equation}\label{eq:plugincond}
\limsup_{x\to\infty}\frac{1-F_{\|\widetilde X\|\mid Y=0}(x)}{1-F_{W\mid Y=0}(x)}<\infty
\quad\text{and}\quad
\|\widehat\delta-\delta\|\frac{a_{n_0}+|b_{n_0}|}{a_{n_0}}=o_p(1).
\end{equation}
Then for any fixed $k\ge3$ the self-normalized largest $k$ order statistics of $\widehat W_i$ in the subsample $Y_i=0$ converge in distribution to \eqref{eq:vstar}, and the conclusion of Theorem \ref{thm:size} holds with $\widehat W_i$ in place of $X_i$.
The second condition in \eqref{eq:plugincond} holds under $\|\widehat\delta-\delta\|=o_p(1)$ when $\gamma>0$, since $b_{n_0}$ may then be taken to be zero, and under $\|\widehat\delta-\delta\|=o_p(1/\log n)$ when $F_{W\mid Y=0}$ has a normal or logistic tail, for which $|b_{n_0}|/a_{n_0}=O(\log n)$.
\end{proposition}

\begin{proof}
Write $\Delta_{n_0}=\max_{i_0\le n_0}|\widehat W^{(0)}_{i_0}-W^{(0)}_{i_0}|$.
The Cauchy-Schwarz inequality gives $\Delta_{n_0}\le\|\widehat\delta-\delta\|\max_{i_0\le n_0}\|\widetilde X^{(0)}_{i_0}\|$.
Let $C$ exceed the limit superior in \eqref{eq:plugincond} and let $s>0$.
A union bound gives
\[
\Prob\left(\max_{i_0\le n_0}\|\widetilde X^{(0)}_{i_0}\|>a_{n_0}s+b_{n_0}\ \middle|\ n_0\right)\le n_0C\left\{1-F_{W\mid Y=0}(a_{n_0}s+b_{n_0})\right\},
\]
which converges to $C(1+\gamma s)^{-1/\gamma}$, read as $Ce^{-s}$ when $\gamma=0$, by the definition of $a_{n_0}$ and $b_{n_0}$, and tends to zero as $s\to\infty$.
Hence $\max_{i_0\le n_0}\|\widetilde X^{(0)}_{i_0}\|=O_p(a_{n_0}+|b_{n_0}|)$ and, by the second condition in \eqref{eq:plugincond}, $\Delta_{n_0}/a_{n_0}=o_p(1)$.

Order statistics are $1$-Lipschitz in the supremum norm, so the top $k$ order statistics of $\widehat W_i$ and of $W_i$ differ by at most $\Delta_{n_0}$, hence by $\Delta_{n_0}/a_{n_0}=o_p(1)$ after division by $a_{n_0}$.
Slutsky's theorem with the continuous mapping theorem applied to \eqref{eq:Xstar} transfers the limit.
The argument of the proof of Theorem \ref{thm:size} then applies unchanged.
\end{proof}
\end{document}